\documentclass [11pt,a4paper]{elsarticle}
\usepackage{float}
\usepackage{amsthm}
\usepackage{amssymb}
\usepackage{bbm}
\usepackage{pifont}
\usepackage{bbding}
\usepackage{xcolor}
\usepackage{mathrsfs}
\usepackage{epsfig,amsmath,amssymb,graphicx}
\usepackage{graphics}
\usepackage{epstopdf}
\usepackage{geometry}
\usepackage[colorlinks,linkcolor=blue,anchorcolor=blue,citecolor=blue]{hyperref}
\usepackage{caption}
\usepackage{graphicx}
\usepackage{subfigure}
\newtheorem{theorem}{Theorem}
\newtheorem{lemma}{Lemma}
\newtheorem{definition}{Definition}

\usepackage{hyperref}

\journal{Elsevier journals}

\makeatletter
\def\ps@pprintTitle{%
  \let\@oddhead\@empty
  \let\@evenhead\@empty
  \def\@oddfoot{\reset@font\hfil\thepage\hfil}
  \let\@evenfoot\@oddfoot
}
\makeatother

\begin{document}

\begin{frontmatter}

\title{Pattern formation and bifurcation analysis of delay induced fractional-order epidemic spreading on networks}


\author[mymainaddress,mysecondaryaddress]{Jiaying Zhou}
\ead{jyzhou0513@gmail.com}

\author[mymainaddress,mysecondaryaddress]{Yong Ye}
\ead{yong\_ye1994@163.com}

\author[mysecondaryaddress]{Alex Arenas\corref{cor1}}
\ead{alexandre.arenas@urv.cat}

\author[mysecondaryaddress]{Sergio G\'omez}
\ead{sergio.gomez@urv.cat}

\author[mymainaddress]{Yi Zhao\corref{cor1}}
\cortext[cor1]{Corresponding author}
\ead{zhao.yi@hit.edu.cn}

\address[mymainaddress]{School of Science, Harbin Institute of Technology, Shenzhen, 518055 P.\ R.\ China}
\address[mysecondaryaddress]{Departament d'Enginyeria Inform{\`a}tica i Matem{\`a}tiques, Universitat Rovira i Virgili, 43007 Tarragona, Spain}

\begin{abstract}

The spontaneous emergence of ordered structures, known as Turing patterns, in complex networks is a phenomenon that holds potential applications across diverse scientific fields, including biology, chemistry, and physics. Here, we present a novel delayed fractional-order susceptible-infected-recovered-susceptible (SIRS) reaction-diffusion model functioning on a network, which is typically used to simulate disease transmission but can also model rumor propagation in social contexts.
Our theoretical analysis establishes the Turing instability resulting from delay, and we support our conclusions through numerical experiments. We identify the unique impacts of delay, average network degree, and diffusion rate on pattern formation. The primary outcomes of our study are: (i) Delays cause system instability, mainly evidenced by periodic temporal fluctuations; (ii)
The average network degree produces periodic oscillatory states in uneven spatial distributions; (iii) The combined influence of diffusion rate and delay results in irregular oscillations in both time and space. However, we also find that fractional-order can suppress the formation of spatiotemporal patterns. These findings are crucial for comprehending the impact of network structure on the dynamics of fractional-order systems.
\end{abstract}


\end{frontmatter}


\section{Introduction}\label{section1}

The study of reaction-diffusion system patterns has been a central focus in research for a long time. The inception of these studies dates back to 1952 when Turing demonstrated that the activator-to-inhibitor diffusion coefficient ratio could cause the destabilization of a steady state, leading to the emergence of periodic spatial patterns~\cite{turing1952chemical}. This phenomenon is now known as the Turing pattern. Turing patterns have been observed in various scenarios, such as autocatalytic chemical reactions with inhibition~\cite{prigogine1968symmetry,castets1990experimental,ouyang1991transition,asllanni22}, epidemic spreading~\cite{sun2012pattern,chang2020cross,chang2022optimal,chang2022sparse,zheng2022pattern,zhou2022complex}, and even ecology~\cite{fernandes2012turing,zhang2014delay,liu2019pattern}. Othmer and Scriven, as early as 1971, highlighted that Turing instability might occur in networked systems and play a significant role in the initial stages of biological morphogenesis, as it spreads through the network connections between cells~\cite{othmer1971instability}. They proposed a general mathematical framework to analyze network instability and further investigated it~\cite{othmer1974non,horsthemke2004network,moore2005localized}, leading to a series of related works~\cite{petit2017theory,zheng2020turing,muolo2023turing}. For instance, in 2010, Nakao and Mikhailov studied Turing patterns in large random networks and observed multiple steady-state coexistences and hysteresis effects~\cite{nakao2010turing}.
Especially during the spread of epidemics, the diffusion of pathogens (similar substances) from high-density spatial regions to low-density spatial regions has led to the development of recognizable spatial explicit models. The research results of pattern dynamics can reveal the distribution structure of populations after spatial diffusion. This enables people to effectively utilize and control population resources. Additionally, these findings provide the scientific basis for preventing and controlling infectious diseases~\cite{sun2016pattern,stancevic2013turing}.

Delays are a widespread phenomenon in natural environments. They can be observed in the gestation period of animals, or in disease transmission models, where delays arise from latent periods or healing cycles, leading to periodic disease outbreaks~\cite{ye2021bifurcation,ye2022promotion,zhou2022bifurcation}.
Subsequently, several studies proposed the use of fractional derivative equations to establish mathematical models for predicting COVID-19~\cite{lu2020fractional,kilicman2018fractional,higazy2020novel}. For example, in 2020, Zhang et al. demonstrated that impacts of death and human activities on nonlocal memory could be captured through a time-fractional derivative equation, contributing to our understanding of COVID-19's death and remission rates~\cite{zhang2020applicability}.
In the same year, Xu and colleagues proposed an improved fractional order SEIQRP model. When tested with epidemic data from the United States, this model successfully predicted short-term epidemic trends. Their results showed that the model effectively characterized the process of disease transmission, providing a theoretical basis for understanding the epidemic~\cite{xu2020forecast}.
Given the universality of delays, in 2019, Chang and his colleagues examined delay-induced Turing patterns using the modified Leslie-Gower model. They analyzed pattern formation in various networks~\cite{chang2019delay}. The following year, they studied Turing patterns on multiplex networks with both self-diffusion and cross-diffusion. Their research resulted in the discovery of heterogeneous patterns exhibiting rich characteristics~\cite{gao2020cross}.

Since the life cycle incorporates memory, fractional calculus equations have been employed to study system dynamics, as integral-order equations cannot account for this inherent memory~\cite{du2013measuring,djilali2020turing,zhang2022impact}. Therefore, in 2022, Zheng et al.\ explored Turing patterns of a fractional-order system on a random network based on the SIR model, discovering that delay and diffusion coefficients influence pattern generation~\cite{zheng2022turing}. However, they used a small random network built with a certain probability, which could not effectively reveal the impact of the network's average degree on pattern formation.

Motivated by Nakao and Mikhailov's work~\cite{nakao2010turing}, we aim to conduct research based on Erd\H{o}s–R\'enyi (ER) random networks to reflect the average network degree's influence on pattern formation, which has been well-established in integer-order systems. To the best of our knowledge, there are limited frameworks that study the effects of delay, diffusion coefficient, time-fractional order, and network average degree on Turing patterns in a delay time-fractional order system. Consequently, this paper plans to introduce factors like delay, diffusion coefficient, and network average degree based on a simple SIRS model, and further investigate whether the time-fractional order affects the uniform stationary state of space, considering diffusion terms in the three-component system as in~\cite{kuznetsov2022robust}.

The delay-induced time-fractional SIRS equations are formulated as follows:
\begin{equation} \label{sys1}
\left\{\begin{aligned}
&D^{q} S_i(t)=\Lambda-\beta S_i(t-\tau) I_i(t-\tau)-\mu S_i(t)+\nu R_i(t)+d_1 \sum_{j=1}^N A_{i j} \left(S_j-S_i\right), \\
&D^{q} I_i(t)=\beta S_i(t-\tau) I_i(t-\tau)-(\gamma+\mu+\alpha) I_i(t)+d_2 \sum_{j=1}^N A_{i j} \left(I_j-I_i\right), \\
&D^{q} R_i(t)=\gamma I_i(t)-(\mu+\nu) R_i(t),\\
&S_i(0)=u_i(t),~I_i(0)=v_i(t),~R_i(0)=w_i(t),
\end{aligned}\right.
\end{equation}
where $D^q$ is the Caputo derivative and $q \in (0, 1]$ is the order of the differential operator. $S_i$, $I_i$ and $R_i$ represent the density of $S$ (susceptible), $I$ (infected)  and $R$ (recovered) in node $i$. Disease transmission (reaction term) occurs inside the node. Concurrently, the diffusive flux of the susceptible $S$ or infected $I$ to node $i$ is the diffusion term, which is expressed as $\sum_{j=1}^N A_{i j}\left(S_j-S_i\right)$ or $\sum_{j=1}^N A_{i j}\left(I_j-I_i\right)$, where $i,j\in \{1,2,\ldots,N\}$. Here, $A_{ij}$ is one if nodes $i$ and $j$ are connected, zero otherwise, i.e., $A$ is the adjacency matrix of the diffusion network. We suppose this network is undirected, thus $A$ is symmetric. Note that, for simplicity, we have considered that there is no diffusion for recovered individuals.
 Other parameters carry the following biological significance: $\Lambda$ indicates the birth rate of $S$, $\beta$ is the transmission rate between susceptible and infected populations, $\mu$ represents the natural mortality rate of populations $S$, $I$, and $R$, $\nu$ is the ratio at which the recovered population returns to the susceptible compartments without acquiring immunity, $\gamma$ denotes the recovery rate of infected individuals, $\alpha$ is the disease-related death rate, $\tau$ corresponds to the disease's latent period, and $d_1$ and $d_2$ represent the self-diffusion coefficients of susceptible and infected, respectively.

The remainder of this paper is organized as follows. In Sec.~\ref{section2}, we present the basic definition and stability lemma for fractional differential equations. In Sec.~\ref{section3}, we theoretically prove the stability of the model without delay and subsequently analyze the Turing instability condition induced by delay. In Sec.~\ref{section4}, we conduct relevant numerical experiments to validate the theoretical findings from previous sections and examine the effects of network average degree, delay, diffusion coefficient, and fractional order on the spatiotemporal pattern. Finally, we discuss the results of our analysis and provide an outlook for future work in Sec.~\ref{section5}.

\section{Preliminaries}\label{section2}

The Caputo fractional derivative is widely used in engineering applications due to its convenience. Therefore, we provide the definition of the Caputo fractional derivative and some essential lemmas for analyzing the stability of fractional-order systems as follows:
\begin{definition} \label{definition1}\cite{podlubny1999199}
The Caputo fractional-order derivative is defined as
\begin{equation*}
_{t_0}^CD_t^q f(t)=\frac{1}{\Gamma(n-q)}\int_{t_0}^t\frac{f^{(n)}(\tau)}{(t-\tau)^{q+1-n}}d\tau,
\end{equation*}
where $q\in (n-1,n)$ and $\Gamma(\cdot)$ is Gamma function.
In particular, we have
 \begin{equation*}
 _{t_0}^CD_t^qf(t)=\frac{1}{\Gamma(1-q)}\int_{t_0}^t\frac{f'(\tau)}{(t-\tau)^{q}}d\tau.
\end{equation*}
when $q\in (0,1)$.
For convenience, we denote $_{t_0}^CD_t^q f(t)$ as $D^q f(t)$.
\end{definition}
\begin{lemma}\cite{matignon1996stability}\label{lemma1}
Consider the fractional-order system
$$
D_t^q x(t)=f(t, x(t))
$$
with initial condition $x\left(t_0\right)=x_{t_0}$, where $q \in(0,1]$. The equilibrium points
are locally asymptotically stable if all eigenvalues $\lambda_i$ of the Jacobian matrix $\frac{\partial f(t, x)}{\partial x}$ calculated at them satisfy
$\left|\arg \left(\lambda_i\right)\right|> \frac{q\pi}{2}$.
\end{lemma}
\begin{lemma}\cite{deng2007stability}\label{lemma2}
Consider the following $n$-dimensional linear fractional-order system with time delay
\begin{equation}\label{deng}
\left\{\begin{array}{c}
D^{q_1} \mathbf{w}_1(t)=\varpi_{11} \mathbf{w}_1\left(t-\tau_{11}\right)+\varpi_{12} \mathbf{w}_2\left(t-\tau_{12}\right)+\cdots+\varpi_{1 n} \mathbf{w}_n\left(t-\tau_{1 n}\right), \\
D^{q_2} \mathbf{w}_2(t)=\varpi_{21} \mathbf{w}_1\left(t-\tau_{21}\right)+\varpi_{22} \mathbf{w}_2\left(t-\tau_{22}\right)+\cdots+\varpi_{2 n} \mathbf{w}_n\left(t-\tau_{2 n}\right), \\
\vdots \\
D^{q_n} \mathbf{w}_n(t)=\varpi_{n 1} \mathbf{w}_1\left(t-\tau_{n 1}\right)+\varpi_{n 2} \mathbf{w}_2\left(t-\tau_{n 2}\right)+\cdots+\varpi_{n n} \mathbf{w}_n\left(t-\tau_{n n}\right),
\end{array}\right.
\end{equation}
where $q_i \in(0,1)$, $i=1,2, \ldots, n$. In system ~\eqref{deng}, define the time delay matrix $\tau=\left(\tau_{i j}\right) \in\left(\mathbb{R}^{+}\right)^{n \times n}$, the coefficient matrix $\varpi=\left(\varpi_{i j}\right)\in\mathbb{R}^{n \times n}$, and then state variables $\mathbf{w}_i(t), \mathbf{w}_i\left(t-\tau_{i j}\right) \in \mathbb{R}$.
Define
$$
\Delta(\lambda)
=\left[\begin{array}{cccc}
\lambda^{q_1}-\varpi_{11} e^{-\lambda \tau_{11}} & -\varpi_{12} e^{-\lambda \tau_{12}} & \cdots & -\varpi_{1 n} e^{-\lambda \tau_{1 n}} \\
-\varpi_{21} e^{-\lambda \tau_{21}} & \lambda^{q_2}-\varpi_{22} e^{-\lambda \tau_{22}} & \cdots & -\varpi_{2 n} e^{-\lambda \tau_{2 n}} \\
\vdots & \vdots & \ddots & \vdots \\
-\varpi_{n 1} e^{-\lambda \tau_{n 1}} & -\varpi_{n 2} e^{-\lambda \tau_{n 2}} & \cdots & \lambda^{q_n}-\varpi_{n n} e^{-\lambda \tau_{n n}}
\end{array}\right] .
$$
Then the zero solution of system~\eqref{deng} is Lyapunov globally asymptotically stable if all the roots of the characteristic equation $\operatorname{det}(\Delta(\lambda))=0$ have negative real parts.
\end{lemma}

\section{Results}\label{section3}

In this section, we primarily focus on the Turing instability of system~\eqref{sys1}. Utilizing the Turing stability theory for delayed reaction-diffusion models in continuous media, it is crucial to ensure that the endemic equilibrium of system~\eqref{sys1} is locally stable in the absence of diffusion and delay. To achieve this, we first need to investigate the stability of endemic equilibrium in the corresponding ordinary differential model.

\subsection{Stability analysis of the dynamic without diffusion and delay}\label{non-diffusion and non-delay}

The equilibrium of system \eqref{sys1} can be derived as follows
\begin{equation} \label{sys2}
\left\{\begin{aligned}
&\Lambda-\beta S_*I_*-\mu S_*+\nu R_*=0, \\
&\beta S_* I_*-(\gamma+\mu+\alpha) I_*=0, \\
&\gamma I_*-(\mu+\nu) R_*=0.
\end{aligned}\right.
\end{equation}
So we have the endemic equilibrium $E^{*}=\left(S_*, I_*, R_*\right)$, where $S_*=\frac{\gamma+\mu+\alpha}{\beta}$, $I_*=\frac{(\mu+\nu)[\beta\Lambda-\mu(\gamma+\mu+\alpha)]}{\beta(\gamma+\mu+\alpha)\mu+\beta(\mu+\alpha)\nu  }$, $R_*=\frac{\beta\gamma\Lambda-\mu\gamma(\gamma+\mu+\alpha)}{\beta(\gamma+\mu+\alpha)\mu+\beta(\mu+\alpha)\nu }$.
In addition, system \eqref{sys1} has  the  disease-free equilibrium $E^{0}=\left(\frac{\Lambda}{\mu}, 0, 0\right)$.
We mainly study the situation
that diseases appear in the initial state, so in this article we do not consider the disease-free equilibrium $E^{0}$.

\begin{theorem}\label{theorem1}
When $\beta<\beta_c$, the disease-free equilibrium $E^ {0}$ of system \eqref{sys1} is locally asymptotically stable for all $\tau \geqslant 0$.
\end{theorem}

\begin{proof} The characteristic matrix of system \eqref{sys1} at the disease-free equilibrium $E^0$ is
$$
\Delta(\lambda)=\left(\begin{array}{ccc}
\lambda^q+\mu &\frac{\beta \Lambda e^{-\lambda \tau}}{\mu} &-\nu \\
0& \lambda^q+(\gamma+\mu+\alpha)- \frac{\beta \Lambda e^{-\lambda \tau}}{\mu} &0\\
0 &-\gamma &\lambda^q+(\mu+\nu)
\end{array}\right).
$$
The characteristic equation at $E^0$ is
\begin{equation} \label{det1}
\operatorname{det}(\Delta(\lambda)) =\left(\lambda^q+\mu\right)\left(\lambda^q+\gamma+\mu+\alpha- \frac{\beta \Lambda e^{-\lambda \tau}}{\mu}\right)\left( \lambda^q+\mu+\nu\right)=0.
\end{equation}
When $\tau=0$, let $s=\lambda^q$, Eq. \eqref{det1} can be rewritten as
$$
\left(s+\mu\right)\left(s+\gamma+\mu+\alpha-\frac{\beta\Lambda}{\mu} \right)\left(s+\mu+\nu\right)=0.
$$
Hence, the eigenvalues are $s_1=-\mu$, $s_2=\frac{\Lambda}{\mu}(\beta-\beta_c)$ and $s_3=-\mu-\nu$, where $\beta_c=\frac{\mu(\gamma+\mu+\alpha)}{\Lambda}$. Obviously, $\left|\arg \left(s_{1,2,3}\right)\right|>\frac{q\pi}{2}$ if $\beta<\beta_c$. It follows from Lemma  \ref{lemma1} that the disease-free equilibrium  $E^0$ is locally asymptotically stable if $\beta<\beta_c$.
When $\tau \neq 0$, the eigenvalues in the first and last terms of Eq.~\eqref{det1} are obviously negative. Therefore, we only need to analyze the second term of Eq.~\eqref{det1},
\begin{equation}\label{det2}
\lambda^q+\gamma+\mu+\alpha- \frac{\beta \Lambda e^{-\lambda \tau}}{\mu}=0 .
\end{equation}
Substituting $\lambda=i \omega$, $(w>0)$, into Eq.~\eqref{det2} we have
$$
(i \omega)^q+\gamma+\mu+\alpha- \frac{\beta \Lambda e^{-i \omega \tau}}{\mu}=0,
$$
which is equivalent to
\begin{equation}\label{det3}
 \omega^q\left(\cos \left(\frac{q \pi}{2}\right) +i \sin\left( \frac{q \pi}{2}\right) \right)  +\gamma+\mu+\alpha- \frac{\beta \Lambda}{\mu}\left(\cos (\omega\tau)-i \sin (\omega\tau)\right)=0.
\end{equation}
Separating the real and imaginary parts of Eq.~\eqref{det3}, one obtains
\begin{equation} \label{det4}
\left\{\begin{aligned}
&\omega^q\cos \left(\frac{q \pi}{2}\right)+\gamma+\mu+\alpha= \frac{\beta \Lambda}{\mu} \cos (\omega\tau), \\
&\omega^q \sin \left(\frac{q \pi}{2}\right)=- \frac{\beta \Lambda}{\mu} \sin(\omega\tau).
\end{aligned}\right.
\end{equation}
By adding the squares of the left and right sides of the two equations in Eq.~\eqref{det4}, we get
\begin{equation}\label{det5}
\omega^{2 q} +2(\gamma+\mu+\alpha) \cos \left(\frac{q \pi}{2}\right) \omega^q  +\left(\gamma+\mu+\alpha+\frac{\beta \Lambda}{\mu} \right)\left(\gamma+\mu+\alpha-\frac{\beta \Lambda}{\mu} \right)=0 .
\end{equation}
If $\beta<\beta_c$, Eq.~\eqref{det5} has no positive roots. Thus, Eq.~\eqref{det5} has no pure imaginary roots. Hence, one obtains $\left|\arg \left(\omega_{1,2,3}^q\right)\right|>\frac{q \pi}{2}$. According to Lemma \ref{lemma1}, $E^0$ is locally asymptotically stable.
\end{proof}
As the diffusion term analysis in the subsequent section involves the stability of the endemic equilibrium $E^*$  when $\tau\neq0$, here we will limit our discussion to the stability of the endemic equilibrium $E^*$  when $\tau=0$. The characteristic matrix of system \eqref{sys1} at the endemic equilibrium $E^*$ if $\tau=0$ is
$$
\Delta(\lambda)=\left(\begin{array}{ccc}
\lambda^q+\mu+\beta  I_* & \beta S_* &-\nu \\
-\beta {I}_*  & \lambda^q+(\gamma+\mu+\alpha)-\beta  {S}_* &0\\
0 &-\gamma &\lambda^q+(\mu+\nu)
\end{array}\right).
$$
The characteristic equation at $E^*$ is
\begin{equation}\label{det6}
\begin{aligned}
\operatorname{det}(\Delta(\lambda))=&\left(\lambda^q+\mu+\beta  I_*\right)\left(\lambda^q+\gamma+\mu+\alpha- \beta  S_*\right)\left( \lambda^q+\mu+\nu\right)\\
&+\beta  I_*\left[ \beta  S_*(\lambda^q+\mu+\nu)-\gamma\nu\right].\\
\end{aligned}
\end{equation}
Setting $s=\lambda^q$, Eq.~\eqref{det6}
can be rewritten as
\begin{align*}
s^3+As^2+Bs+C=0,
\end{align*}
where $A=2\mu+\nu+\beta I_*$, $B=\mu(\mu+\nu)+(2\mu+\alpha+\nu+\gamma)\beta I_*$, $C=[(\mu+\alpha)(\mu+\nu)+\gamma\mu]\beta I_*$.
According to Routh-Hurwitz criterion, when $A>0$, $B>0$, $C>0$, $AB>C$, $E^*$ is locally asymptotically
stable. Obviously, $A>0$, $B>0$, $C>0$, $AB>C$, if $I^*>0$. That is to say, $E^*$ is locally asymptotically stable if $\beta>\beta_c$.
\subsection{Turing instability induced by delay}\label{diffusion and delay}

In Sec.~\ref{non-diffusion and non-delay}, we have given the conditions for the stability of endemic disease equilibrium, and on this basis, this section plans to study the impact of delay on the stability of system~\eqref {sys1}. Therefore, the linearization form of the fractional-order system~\eqref {sys1} with a delay at the endemic disease equilibrium $E^*$ is rewritten as follows
\begin{equation} \label{sys3}
\left\{\begin{aligned}
&D^{q} S_i(t)=-\mu S_i(t)+\nu R_i(t)-\beta {I}_*S_i(t-\tau)-\beta {S}_*I_i(t-\tau)+d_1 \sum_{j=1}^n L_{i j} S_j, \\
&D^{q} I_i(t)=-(\gamma+\mu+\alpha) I_i(t)+\beta {I}_*S_i(t-\tau)+\beta {S}_*I_i(t-\tau)+d_2 \sum_{j=1}^n L_{i j}  I_j, \\
&D^{q} R_i(t)=\gamma I_i(t)-(\mu+\nu) R_i(t),
\end{aligned}\right.
\end{equation}
where $L_{ij}$ are the components of the graph Laplacian $L$ corresponding to the diffusion graph with components $A_{ij}$, i.e., $L=K-A$, where $K$ is the diagonal matrix with the degrees of the nodes.
By applying the Laplace transform to both sides of system \eqref{sys3}, we obtain the following:
\begin{equation}
\left\{\begin{aligned}\label{sys4}
\lambda^q X_i-\lambda^{q-1} u_i(0)=&-\mu X_i+\nu Z_i-\beta  I_*e^{-\lambda \tau}\left(X_i+\int_{-\tau}^0 e^{-\lambda t} u_i(t) d t\right) \\
&-\beta S_* e^{-\lambda \tau}\left(Y_i+\int_{-\tau}^0 e^{-\lambda t} v_i(t) d t\right)+d_1 \sum_{j=1}^n L_{i j} X_j, \\
\lambda^q Y_i-\lambda^{q-1} v_i(0)=&-(\gamma+\mu+\alpha) Y_i+\beta {I}_*  e^{-\lambda \tau}\left(X_i+\int_{-\tau}^0 e^{-\lambda t} u_i(t) d t\right) \\
&+\beta  {S}_* e^{-\lambda \tau}\left(Y_i+\int_{-\tau}^0 e^{-\lambda t} v_i(t) d t\right)+d_2 \sum_{j=1}^n L_{i j} Y_j,\\
\lambda^q Z_i-\lambda^{q-1} w_i(0)=&\gamma Y_i -(\mu+\nu)Z_i,
\end{aligned}\right.
\end{equation}
where $X_i, Y_i, Z_i$ is the Laplace transform of $S_i, I_i, R_i$, respectively.
System \eqref{sys4} can be reformulated in the following matrix form:
\begin{equation} \label{sys5}
\left(A-D L_1\right) X=b,
\end{equation}
where
$$
A=\left(\begin{array}{ccc}
\lambda^q+\mu+\beta  I_*e^{-\lambda \tau} & \beta S_* e^{-\lambda \tau}&-\nu \\
-\beta {I}_*  e^{-\lambda \tau}& \lambda^q+(\gamma+\mu+\alpha)-\beta  {S}_* e^{-\lambda \tau}&0\\
0 &-\gamma &\lambda^q+(\mu+\nu)
\end{array}\right) \otimes E ,
$$
\begin{align*}
D=\left(\begin{array}{lll}d_1 & 0 &0 \\ 0 & d_2 & 0\\ 0 & 0& 0 \end{array}\right) \otimes E,
\end{align*}
\begin{align*}
X=\left(X_1, X_2, \ldots, X_n, Y_1, Y_2, \ldots, Y_n, Z_1, Z_2, \ldots, Z_n\right)^T,\\
b=\left(b_{11}, b_{12}, \ldots b_{1 n}, b_{21}, b_{22}, \ldots b_{2 n}, b_{31}, b_{32}, \ldots b_{3 n}\right)^T,
\end{align*}
$$
\begin{gathered}
\left(\begin{array}{l}
b_{1 i} \\
b_{2 i} \\
b_{3 i}
\end{array}\right)=\left(\begin{array}{c}
\lambda^{q-1} u_i(0)-\beta  I_*e^{-\lambda \tau}  \int_{-\tau}^0 e^{-\lambda t} u_i(t) d t-\beta  S_*e^{-\lambda \tau} \int_{-\tau}^0 e^{-\lambda t} v_i(t) d t \\
\lambda^{q-1} v_i(0)+\beta  I_*e^{-\lambda \tau} \int_{-\tau}^0 e^{-\lambda t} u_i(t) d t+\beta  I_*e^{-\lambda \tau} \int_{-\tau}^0 e^{-\lambda t} v_i(t) d t \\
\lambda^{q-1} w_i(0)
\end{array}\right), \\
L_1=\left(\begin{array}{ccc}
L & 0 &0\\
0 & L &0\\
0 & 0 &L
\end{array}\right),
\end{gathered}
$$
matrix $E$ is an $n\times n$ identity matrix and $\otimes$ is kronecker product. $A-D L_1$ represents the characteristic matrix of system \eqref{sys3}.

Since the Laplacian matrix is a real symmetric matrix, it can be diagonalized. An orthonormal basis $\phi_i$ makes the following equation hold:
$$
L_1 \phi=\Lambda \phi\,,
$$
where $\Lambda_i$ is the eigenvalue of $L,  \phi=\left(\phi_1, \ldots, \phi_n, \phi_1, \ldots, \phi_n, \phi_1, \ldots, \phi_n\right)^T$ is an invertible matrix, $\phi_i$ is the eigenvector of $\Lambda_i$, and
$$
\begin{aligned}
\Lambda &=\left(\begin{array}{lll}
\Lambda^{(1)} & 0 & 0 \\
0 & \Lambda^{(1)}& 0\\
0 & 0 &\Lambda^{(1)}
\end{array}\right), \\
\Lambda^{(1)} &=\left(\begin{array}{cccc}
\Lambda_1 & 0 & \ldots & 0 \\
0 & \Lambda_2 & \ldots & 0 \\
0 & \ldots & \ldots & 0 \\
0 & \ldots & 0 & \Lambda_n
\end{array}\right) .
\end{aligned}
$$
Supposing $X=\phi Y$, system \eqref{sys5} can be rewritten as
$$
A \phi Y-D L_1 \phi Y=b \Rightarrow A \phi Y-D \Lambda \phi Y=b \Rightarrow(A-D \Lambda) X=b .
$$
Thus, system \eqref{sys5} can be reduced to
$$
\left(\begin{array}{ccc}
\lambda^q+\mu+\beta  I_*e^{-\lambda \tau}-d_1 \Lambda_i & \beta S_* e^{-\lambda \tau}&-\nu \\
-\beta {I}_*  e^{-\lambda \tau}& \lambda^q+(\gamma+\mu+\alpha)-\beta  {S}_* e^{-\lambda \tau}-d_2 \Lambda_i&0\\
0 &-\gamma &\lambda^q+(\mu+\nu)
\end{array}\right)\left(\begin{array}{c}
X_i \\
Y_i \\
Z_i
\end{array}\right)=\left(\begin{array}{c}
b_{1 i} \\
b_{2 i} \\
b_{3 i}
\end{array}\right) .
$$
It is widely recognized that initial values do not affect the stability of linear fractional differential systems. Assuming all initial values are zero, the stability of system \eqref{sys1} can be determined by:
$$
\left(\begin{array}{ccc}
\lambda^q+\mu+\beta  I_*e^{-\lambda \tau}-d_1 \Lambda_i & \beta S_* e^{-\lambda \tau}&-\nu \\
-\beta {I}_*  e^{-\lambda \tau}& \lambda^q+(\gamma+\mu+\alpha)-\beta  {S}_* e^{-\lambda \tau}-d_2 \Lambda_i&0\\
0 &-\gamma &\lambda^q+(\mu+\nu)
\end{array}\right)\left(\begin{array}{l}
X_i \\
Y_i \\
Z_i
\end{array}\right)=\left(\begin{array}{l}
0 \\
0 \\
0
\end{array}\right) .
$$
Therefore, the stability of system \eqref{sys1} depends on the following characteristic equation
$$
\left|\begin{array}{ccc}
\lambda^q+\mu+\beta  I_*e^{-\lambda \tau}-d_1 \Lambda_i & \beta S_* e^{-\lambda \tau}&-\nu \\
-\beta {I}_*  e^{-\lambda \tau}& \lambda^q+(\gamma+\mu+\alpha)-\beta  {S}_* e^{-\lambda \tau}-d_2 \Lambda_i&0\\
0 &-\gamma &\lambda^q+(\mu+\nu)
\end{array}\right|=0,
$$
namely,
\begin{equation} \label{sys6}
P_1(\lambda)+P_2(\lambda) e^{-\lambda \tau}=0,
\end{equation}
where
$$
\begin{aligned}
P_1(\lambda)=
&\lambda^{3 q}+(\gamma+3 \mu+\alpha+\nu) \lambda^{2 q}+\left[(\gamma+\mu+\alpha) \mu +(\mu+\nu)(\gamma+2 \mu+\alpha)\right] \lambda^q\\
&-\left(d_1+d_2\right) \Lambda_i \lambda^{2 q}-\left[d_1(\gamma+\mu+\alpha)+\mu d_2+(\mu+\nu)\left(d_1+d_2\right)\right] \Lambda_i \lambda^q \\
&+d_1 d_2 \Lambda_i^2 \lambda^q+(\mu+\nu)(\gamma+\mu+\alpha) \mu\\
&-(\mu+\nu)\left[d_1(\gamma+\mu+\alpha)+\mu d_2\right] \Lambda_i+d_1 d_2(\mu+\nu) \Lambda_i^2,\\
P_2(\lambda)=&\beta\left(I_*-S_*\right) \lambda^{2 q} +\left[\left(\gamma+\mu+\alpha\right) \beta I_*+\beta\left(I_*-S_*\right)(\mu+\nu)-\mu \beta S_*\right] \lambda^q \\
&+\left(d_1 \beta S_*-d_2 \beta I_*\right) \Lambda_i \lambda^q +(\mu+\nu)\left(d_1 \beta S_*-d_2 \beta I_*\right) \Lambda_i\\
&   +\left[(\gamma+\mu+\alpha) \beta I_*-\mu \beta S_*\right](\mu+\nu)-\beta I_* \gamma \nu.\\
\end{aligned}
$$
We substitute $\lambda=i \omega=\omega\left(\cos \left(\frac{\pi}{2}\right)+i \sin \left(\frac{\pi}{2}\right)\right)=\omega e^{i \frac{\pi}{2}}$ into system \eqref{sys6}, and have
\begin{equation} \label{sys7}
\left(A_1+i B_1\right)+\left(A_2+i B_2\right)(\cos (\omega \tau)-i \sin (\omega \tau))=0,
\end{equation}
where
$$
\begin{aligned}
A_1=&\omega^{3 q}\cos (3 / 2 q \pi)+(\gamma+3 \mu+\alpha+\nu) \omega^{2 q}\cos (q \pi)\\
&+\left[(\gamma+\mu+\alpha) \mu +(\mu+\nu)(\gamma+2 \mu+\alpha)\right] \omega^{ q}\cos (1 / 2 q \pi)\\
&-\left(d_1+d_2\right) \Lambda_i \omega^{2 q}\cos (q \pi)+d_1 d_2 \Lambda_i^2 \omega^{ q}\cos (1 / 2 q \pi)\\
&-\left[d_1(\gamma+\mu+\alpha)+\mu d_2+(\mu+\nu)\left(d_1+d_2\right)\right] \Lambda_i \omega^{ q}\cos (1 / 2 q \pi) \\
&+(\mu+\nu)(\gamma+\mu+\alpha) \mu-(\mu+\nu)\left[d_1(\gamma+\mu+\alpha)+\mu d_2\right] \Lambda_i+d_1 d_2(\mu+\nu) \Lambda_i^2,\\
B_1=&\omega^{3 q}\sin (3 / 2 q \pi)+(\gamma+3 \mu+\alpha+\nu) \omega^{2 q}\sin (q \pi)\\
&+\left[(\gamma+\mu+\alpha) \mu +(\mu+\nu)(\gamma+2 \mu+\alpha)\right] \omega^{ q}\sin (1 / 2 q \pi) \\
&-\left(d_1+d_2\right) \Lambda_i \omega^{2 q}\sin (q \pi)+d_1 d_2 \Lambda_i^2 \omega^{ q}\sin (1 / 2 q \pi)\\
&-\left[d_1(\gamma+\mu+\alpha)+\mu d_2+(\mu+\nu)\left(d_1+d_2\right)\right] \Lambda_i \omega^{ q}\sin (1 / 2 q \pi),\\
A_2=&\beta\left(I_*-S_*\right) \omega^{2 q}\cos (q \pi) +\left(d_1 \beta S_*-d_2 \beta I_*\right) \Lambda_i \omega^{ q}\cos (1 / 2 q \pi)\\
&+\left[\left(\gamma+\mu+\alpha\right) \beta I_*+\beta\left(I_*-S_*\right)(\mu+\nu)-\mu \beta S_*\right] \omega^{ q}\cos (1 / 2 q \pi) \\
& +(\mu+\nu)\left(d_1 \beta S_*-d_2 \beta I_*\right) \Lambda_i  +\left[(\gamma+\mu+\alpha) \beta I_*-\mu \beta S_*\right](\mu+\nu)-\beta I_* \gamma \nu, \\
B_2=&\beta\left(I_*-S_*\right) \omega^{2 q}\sin (q \pi)+\left(d_1 \beta S_*-d_2 \beta I_*\right) \Lambda_i \omega^{ q}\sin (1 / 2 q \pi)\\
 &+\left[\left(\gamma+\mu+\alpha\right) \beta I_*+\beta\left(I_*-S_*\right)(\mu+\nu)-\mu \beta S_*\right] \omega^{ q}\sin (1 / 2 q \pi).
\end{aligned}
$$
Separating the real and imaginary parts of Eq.~\eqref{sys7}, one obtains
\begin{equation*}
\left\{\begin{aligned}
&A_2 \cos (\omega \tau)+B_2 \sin (\omega \tau)=-A_1, \\
&-A_2 \sin (\omega \tau)+B_2 \cos (\omega \tau)=-B_1,
\end{aligned}\right.
\end{equation*}
then,
\begin{equation} \label{det7}
\left\{\begin{aligned}
&(A_2^2 + B_2^2) \cos (\omega \tau) =-B_1B_2-A_1A_2, \\
&(A_2^2 + B_2^2) \sin  (\omega \tau) =B_1A_2-A_1B_2.
\end{aligned}\right.
\end{equation}
By adding the squares of the left and right sides of the two equations in Eq.~\eqref{det7}, we get
\begin{equation} \label{sys8}
(A_2^2 + B_2^2)^2=(B_1B_2+A_1A_2)^2+(B_1A_2-A_1B_2)^2,
\end{equation}
where $\omega$ can be solved from system \eqref{det7}.
The critical value of $\tau_c$ is
$$
\tau_c=\min _{i, k}\left\{\frac{1}{\omega_k} \arccos \left(\frac{-B_1B_2-A_1A_2}{A_2{ }^2+B_2{ }^2}\right)+\frac{2 \pi}{\omega_k}\right\},
$$
where index $i$ refers to the $i$th node, and $\omega_k$ represent all the positive roots of system \eqref{det7}.
Also, we have the transversality condition
$$
\frac{d \lambda}{d \tau}=\frac{\lambda P_2(\lambda) e^{-\lambda \tau}}{P_1^{\prime}(\lambda)+P_2^{\prime}(\lambda) e^{-\lambda \tau}-\tau P_2(\lambda) e^{-\lambda \tau}}=\frac{M}{N},
$$
and
$$
\operatorname{Re}\left[\frac{d \lambda}{d \tau}\right]=\frac{M_1 N_1+M_2 N_2}{N_1^2+N_2^2},
$$
where
$$
\begin{aligned}
M_1=&- \beta\left(I_*-S_*\right)\omega^{2q+1} \sin (\pi q) \cos (\omega \tau)+\beta\left(I_*-S_*\right)\omega^{2q+1} \cos (\pi q) \sin (\omega \tau) \\
& -\left[\left(\gamma+\mu+\alpha\right) \beta I_*+\beta\left(I_*-S_*\right)(\mu+\nu)-\mu \beta S_*+\left(d_1 \beta S_*-d_2 \beta I_*\right) \Lambda_i\right]  \omega^{q+1} \sin (1 / 2\pi q)  \cos (\omega \tau)\\
& +\left[\left(\gamma+\mu+\alpha\right) \beta I_*+\beta\left(I_*-S_*\right)(\mu+\nu)-\mu \beta S_*+\left(d_1 \beta S_*-d_2 \beta I_*\right) \Lambda_i\right] \omega^{q+1} \cos (1 / 2\pi q) \sin (\omega \tau) \\
&+\left\{ (\mu+\nu)\left(d_1 \beta S_*-d_2 \beta I_*\right) \Lambda_i +\left[(\gamma+\mu+\alpha) \beta I_*-\mu \beta S_*\right](\mu+\nu)-\beta I_* \gamma \nu\right\} \omega \sin (\omega \tau),\\
M_2=&\beta\left(I_*-S_*\right)\omega^{2q+1} \cos (\pi q) \cos (\omega \tau)+\beta\left(I_*-S_*\right)\omega^{2q+1} \sin (\pi q) \sin (\omega \tau)\\
& +\left[\left(\gamma+\mu+\alpha\right) \beta I_*+\beta\left(I_*-S_*\right)(\mu+\nu)-\mu \beta S_*+\left(d_1 \beta S_*-d_2 \beta I_*\right) \Lambda_i\right] \omega^{q+1} \cos (1 / 2\pi q) \cos (\omega \tau) \\
&+\left\{ (\mu+\nu)\left(d_1 \beta S_*-d_2 \beta I_*\right) \Lambda_i +\left[(\gamma+\mu+\alpha) \beta I_*-\mu \beta S_*\right](\mu+\nu)-\beta I_* \gamma \nu\right\} \omega \cos (\omega \tau)\\
& +\left[\left(\gamma+\mu+\alpha\right) \beta I_*+\beta\left(I_*-S_*\right)(\mu+\nu)+\mu \beta S_*+\left(d_1 \beta S_*-d_2 \beta I_*\right) \Lambda_i\right]  \omega^{q+1} \sin (1 / 2\pi q) \sin (\omega \tau),\\
N_1=&3 \alpha \omega^{3q-1}\sin (3/2\pi q)+2 \alpha[(\gamma+3 \mu+\alpha+\nu)-\left(d_1+d_2\right)\Lambda_i  ]\omega^{2q-1} \sin (\pi \alpha)\\
&+ \alpha\{\left[(\gamma+\mu+\alpha) \mu +(\mu+\nu)(\gamma+2 \mu+\alpha)\right]-\left[d_1(\gamma+\mu+\alpha)+\mu d_2+(\mu+\nu)\left(d_1+d_2\right)\right] \Lambda_i +d_1 d_2 \Lambda_i^2 \}\\
&\times\omega^{q-1}\sin (1/2\pi q)+2 \alpha \beta\left(I_*-S_*\right)\omega^{2q-1}( \sin (\pi q)\cos(\omega \tau)-\cos (\pi q)\sin (\omega \tau))\\
&+ \alpha\left[\left(\gamma+\mu+\alpha\right) \beta I_*+\beta\left(I_*-S_*\right)(\mu+\nu)-\mu \beta S_*+\left(d_1 \beta S_*-d_2 \beta I_*\right) \Lambda_i\right] \\
&\times\omega^{q-1}( \sin (1/2\pi q)\cos(\omega \tau)-\cos (1/2\pi q)\sin (\omega \tau))\\
&-\tau \beta\left(I_*-S_*\right) \omega^{2q}( \cos (\pi q)\cos(\omega \tau)+\sin (\pi q)\sin (\omega \tau))\\
&-\tau\left[\left(\gamma+\mu+\alpha\right) \beta I_*+\beta\left(I_*-S_*\right)(\mu+\nu)-\mu \beta S_*+\left(d_1 \beta S_*-d_2 \beta I_*\right) \Lambda_i\right]\\
&\times\omega^{q}( \cos (1/2\pi q)\cos(\omega \tau)+\sin (1/2\pi q)\sin (\omega \tau))\\
& +\{\tau\beta I_* \gamma \nu-\tau(\mu+\nu)\left(d_1 \beta S_*-d_2 \beta I_*\right) \Lambda_i-\tau\left[(\gamma+\mu+\alpha) \beta I_*-\mu \beta S_*\right](\mu+\nu)\} \cos(\omega \tau),
\end{aligned}
$$
$$
\begin{aligned}
N_2=&-3 \alpha \omega^{3q-1}\cos (3/2\pi q)-2 \alpha[(\gamma+3 \mu+\alpha+\nu)-\left(d_1+d_2\right)\Lambda_i  ]\omega^{2q-1}\cos (\pi q)\\
&- \alpha\{\left[(\gamma+\mu+\alpha) \mu +(\mu+\nu)(\gamma+2 \mu+\alpha)\right]-\left[d_1(\gamma+\mu+\alpha)+\mu d_2+(\mu+\nu)\left(d_1+d_2\right)\right] \Lambda_i +d_1 d_2 \Lambda_i^2 \}\\
&\times\omega^{q-1}\cos (1/2\pi q)-2 \alpha \beta\left(I_*-S_*\right)\omega^{2q-1}( \cos(\pi q)\cos (\omega \tau)+\sin(\pi q)\sin (\omega \tau))\\
&- \alpha\left[\left(\gamma+\mu+\alpha\right) \beta I_*+\beta\left(I_*-S_*\right)(\mu+\nu)-\mu \beta S_*+\left(d_1 \beta S_*-d_2 \beta I_*\right) \Lambda_i\right] \\
&\times\omega^{q-1}(\cos(1/2\pi q)\cos (\omega \tau)+\sin(1/2\pi q)\sin (\omega \tau))\\
&-\tau \beta\left(I_*-S_*\right) \omega^{2q}( \sin(\pi q)\cos (\omega \tau)-\cos(\pi q)\sin (\omega \tau))\\
&-\tau\left[\left(\gamma+\mu+\alpha\right) \beta I_*+\beta\left(I_*-S_*\right)(\mu+\nu)-\mu \beta S_*+\left(d_1 \beta S_*-d_2 \beta I_*\right) \Lambda_i\right]\\
&\times\omega^{q}( \sin(1/2\pi q)\cos (\omega \tau)-\cos(1/2\pi q)\sin (\omega \tau))\\
& -\{\tau\beta I_* \gamma \nu-\tau(\mu+\nu)\left(d_1 \beta S_*-d_2 \beta I_*\right) \Lambda_i-\tau\left[(\gamma+\mu+\alpha) \beta I_*-\mu \beta S_*\right](\mu+\nu)\}\sin (\omega \tau).\\
\end{aligned}
$$
Furthermore,
$$
\begin{aligned}
&\left.M\left(\omega i\right)\right|_{\tau=\tau_c}=M_1+iM_2, \\
&\left.N\left(\omega i\right)\right|_{\tau=\tau_c}=N_1+iN_2,
\end{aligned}
$$
where $M_1, M_2, N_1, N_2$ are the real and imaginary parts of $M(\lambda), N(\lambda)$. We suppose $\tau$ is the control parameter and through simple calculations, it can be concluded that
$$
\operatorname{Re}\left[\frac{d \lambda}{d \tau}\right]_{\tau=\tau_c, \omega=\omega_c}=\frac{M_1 N_1+M_2 N_2}{N_1^2+N_2^2} \neq 0,
$$
where $\omega_c$ is the corresponding frequency of $\tau_c$. Thus, based on the above analysis and Hopf bifurcation theory, one has the following results.
\begin{theorem}\label{theorem} Turing instability induced by delay.
\begin{itemize}
    \item If $\operatorname{Re}\left[\frac{d \lambda}{d \tau}\right]_{\tau=\tau_c, \omega=\omega_c}>0$, Turing instability occurs in system \eqref{sys1} when $\tau>\tau_c$.
    \item If $\operatorname{Re}\left[\frac{d \lambda}{d \tau}\right]_{\tau=\tau_c, \omega=\omega_c}<0$, Turing instability occurs in system \eqref{sys1} when $\tau<\tau_c$.
\end{itemize}
\end{theorem}

\section{Numerical analysis}\label{section4}

In this section, we aim to design several numerical experiments to validate the theoretical analysis. First, we calculate the stability conditions of the endemic equilibrium $E^*$ without the diffusion term (i.e., $\beta>\frac{\mu(\gamma+\mu+\alpha)}{\Lambda}$), meaning that we need to ensure the stability of system~\eqref{sys1} without delay, and then investigate the effect of delay on the system. Consequently, we set the parameter values as $\Lambda=5$, $\mu=0.035$, $\nu=0.05$, $\gamma=0.2$, $\alpha=0.01$, $\beta=0.006$, $d_1=0$, $d_2=0$, and $q=1$. This ensures the stability of the non-delay and non-diffusion model~\eqref{sys1}, as $\beta>\beta_c=0.0017$.

Furthermore, based on this, we find that by suitably increasing the delay value, the model transitions from stability to instability, with the critical delay value being $\tau_c \approx 23.06$ (see Fig.~\ref{bifu}). Moreover, when calculating and examining the induced instability conditions, we discover that the fractional order and diffusion terms considered by the model also play a crucial role. Thus, we also provide the corresponding fractional-order threshold of $q=0.95$ with non-diffusion and $\tau_c \approx 33.32$ (see Fig.~\ref{bifu}). We observe that a larger delay is required to render the model~\eqref{sys1} unstable as the fractional order decreases (see Fig.~\ref{bifu}).

Next, we present three sets of experiments: (1) when $q=0.95$, we choose $\tau_1=30$ and $\tau_2=40$, satisfying $\tau_1<\tau_c=33.32<\tau_2$. The time series plot and bifurcation diagram with $\tau$ as the bifurcation parameter are provided (see Fig.~\ref{time series}(a,b)). (2) When $q=1$, we select $\tau_1=20$ and $\tau_2=30$, satisfying $\tau_1<\tau_c=23.06=\tau_2$. The time series plot and bifurcation diagram with $\tau$ as the bifurcation parameter are provided (see Fig.~\ref{time series}(c,d)). (3) When $\tau=30$, we choose $q_1=0.95$ and $q_2=1$, satisfying $q_1<0.965<q_2$. The time series plot and bifurcation diagram with $q$ as the bifurcation parameter is provided (see Fig.~\ref{time series}(e,f)). Similarly, we also demonstrate the corresponding relationship between the eigenvalue of the Laplace matrix $\Lambda_i$ and the delay threshold $\tau_c$ in the cases of $q=0.95$ and $q=1$ for the fractional order when diffusion is considered (see Fig.~\ref{immunization}).

\begin{figure}[h]
\centering
\includegraphics[height=60mm, width=90mm]{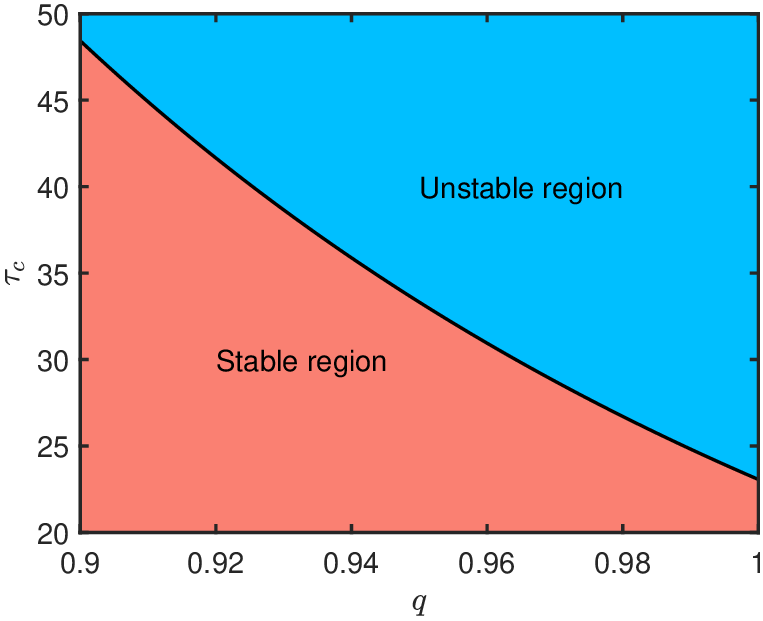}
\caption{The critical value $\tau_c$ decreases as $q$ increases when $\Lambda=5$, $\mu=0.035$, $\nu=0.05$, $\gamma=0.2$, $\alpha=0.01$, $\beta=0.006$, $d_1=0$, and $d_2=0$. In this case, the deep sky blue area represents the unstable region, while the salmon area represents the stable region.}
\label{bifu}
\end{figure}

\begin{figure}[htbp]
\centering
\subfigure[]
{\includegraphics[width=0.42\textwidth ,height=1.7in]{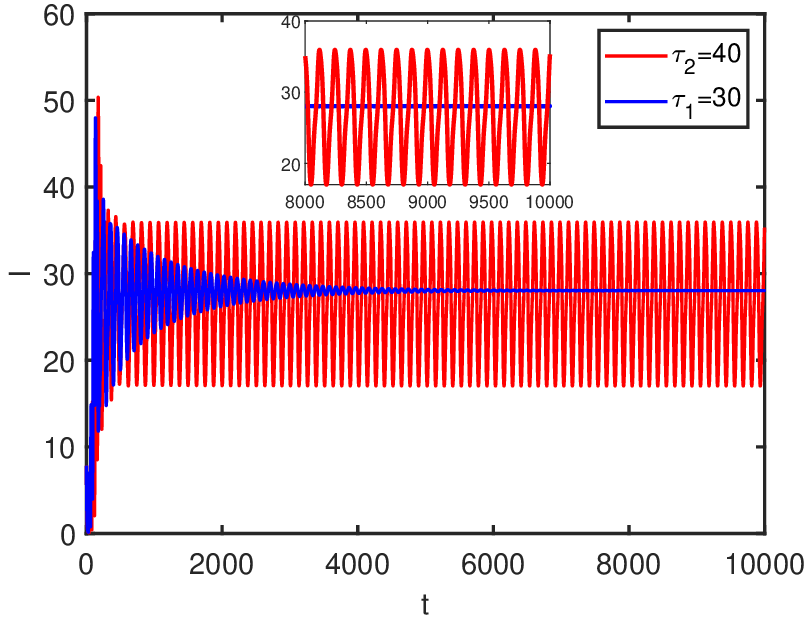}}
\subfigure[]
{\includegraphics[width=0.42\textwidth ,height=1.7in]{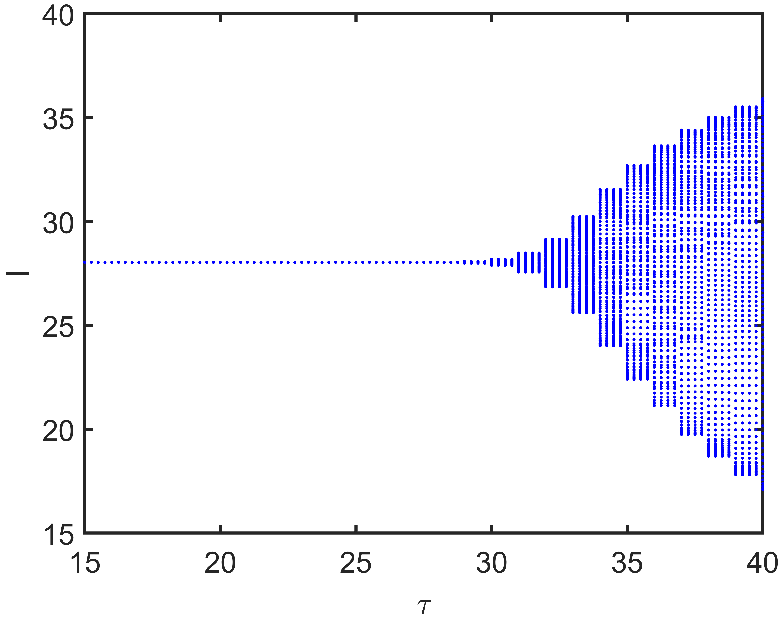}}
\subfigure[]
{\includegraphics[width=0.42\textwidth ,height=1.7in]{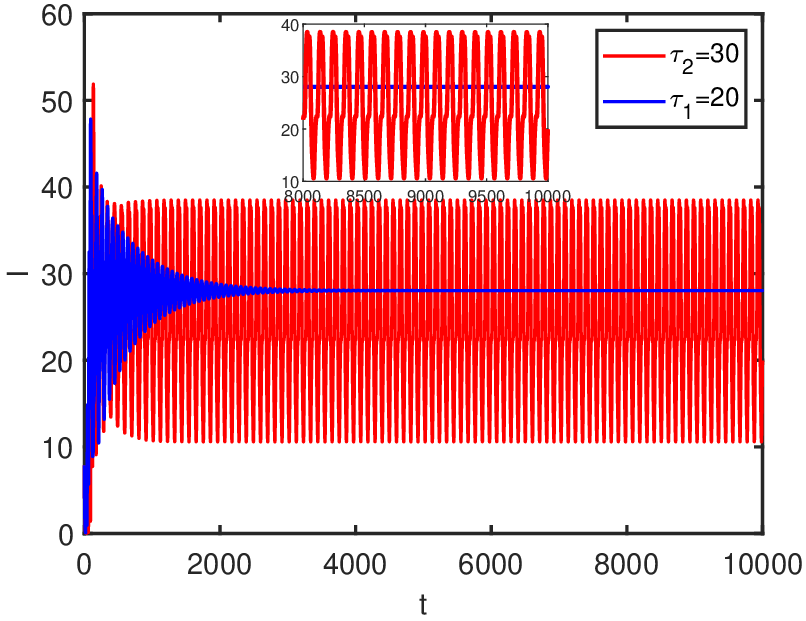}}
\subfigure[]
{\includegraphics[width=0.42\textwidth ,height=1.7in]{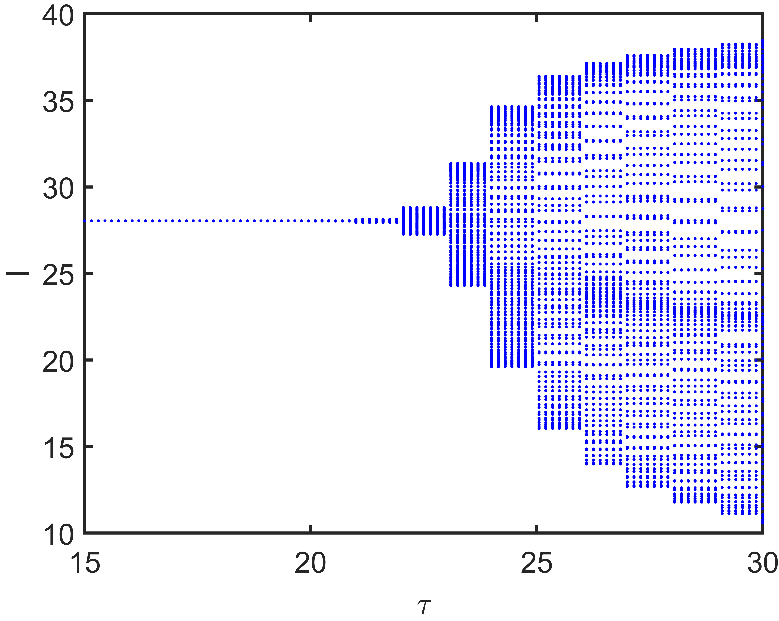}}
\subfigure[]
{\includegraphics[width=0.42\textwidth ,height=1.7in]{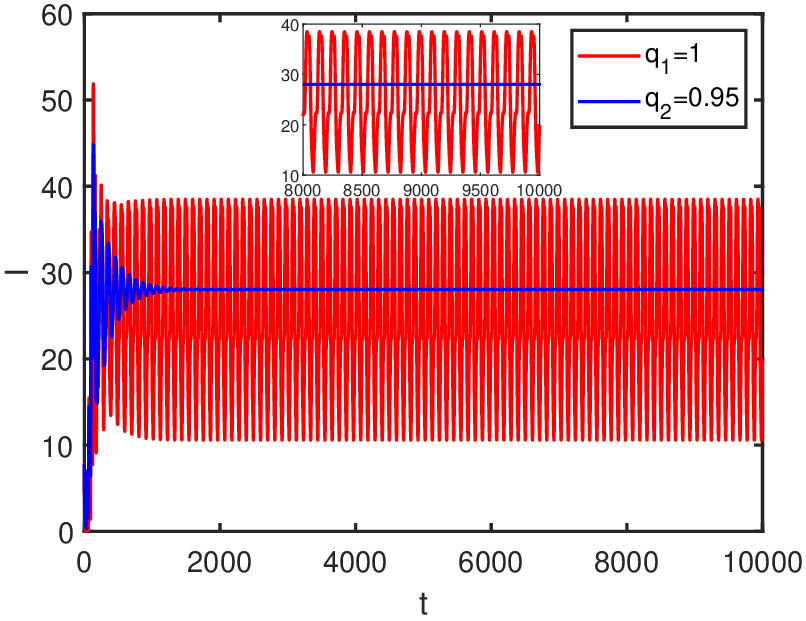}}
\subfigure[]
{\includegraphics[width=0.42\textwidth ,height=1.7in]{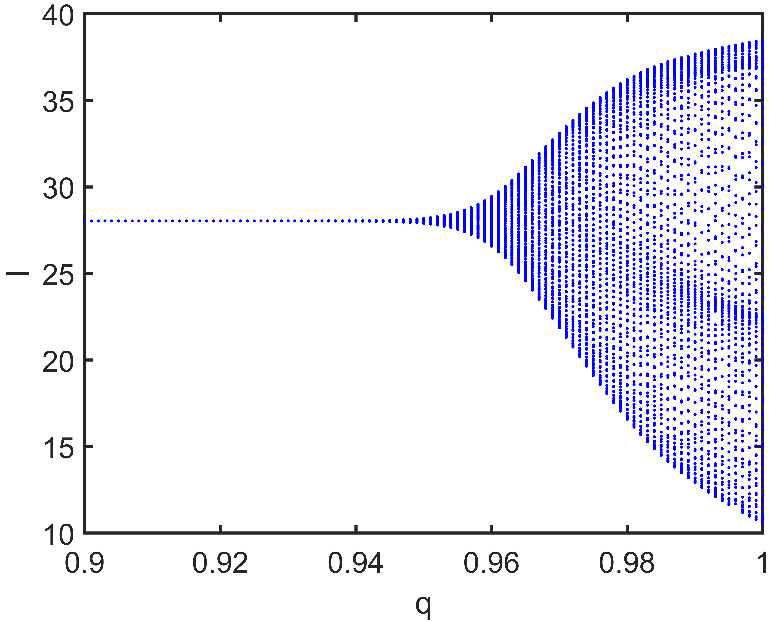}}
\caption{Time series and bifurcation diagram of the system \eqref{sys1} without network when $\Lambda=5$, $\mu=0.035$, $\nu=0.05$, $\gamma=0.2$, $\alpha=0.01$, $\beta=0.006$. (a) $E^{*}$ is stable when $\tau=30$  and periodic when $\tau=40$, $q=0.95$. (b) The Hopf bifurcation occurs about  $\tau$ when $q=0.95$. (c) $E^{*}$ is stable when $\tau=20$ and periodic $\tau=30$ when $q=1$. (d) The Hopf bifurcation occurs about $\tau$ when $q=1$. (e) $E^{*}$ is stable when $q=0.95$ and periodic when $q=1$, $\tau=30$. (f) The Hopf bifurcation occurs about $q$ when $\tau=30$. The red curve represents the periodic solution, and the blue curve represents the stable endemic disease equilibrium.
}\label{time series}
\end{figure}

\begin{figure}[h]
\centering
\subfigure[The critical value $\tau_c$ about  $\Lambda_i$ when $q=0.95$.]{
\includegraphics[width=0.45\textwidth ,height=2in]{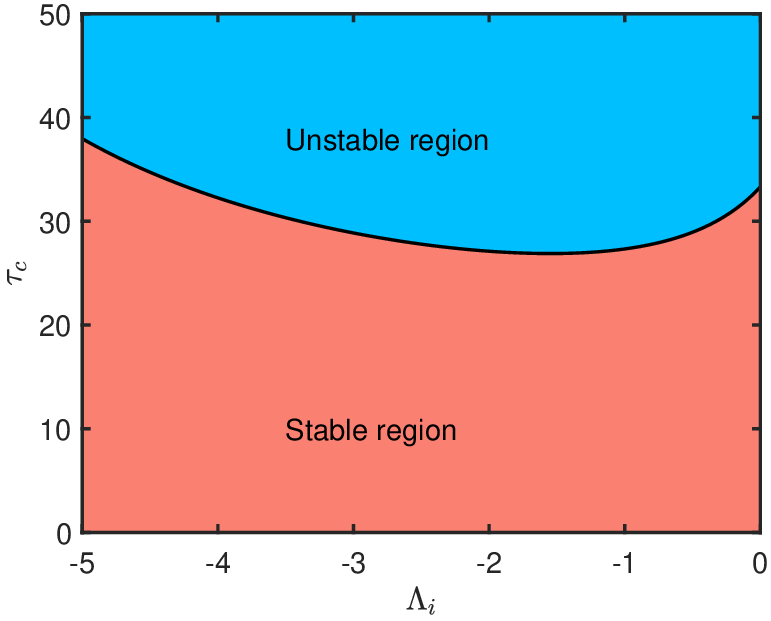}}
\subfigure[The critical value $\tau_c$ about  $\Lambda_i$ when $q=1$.]{
\includegraphics[width=0.45\textwidth ,height=2in]{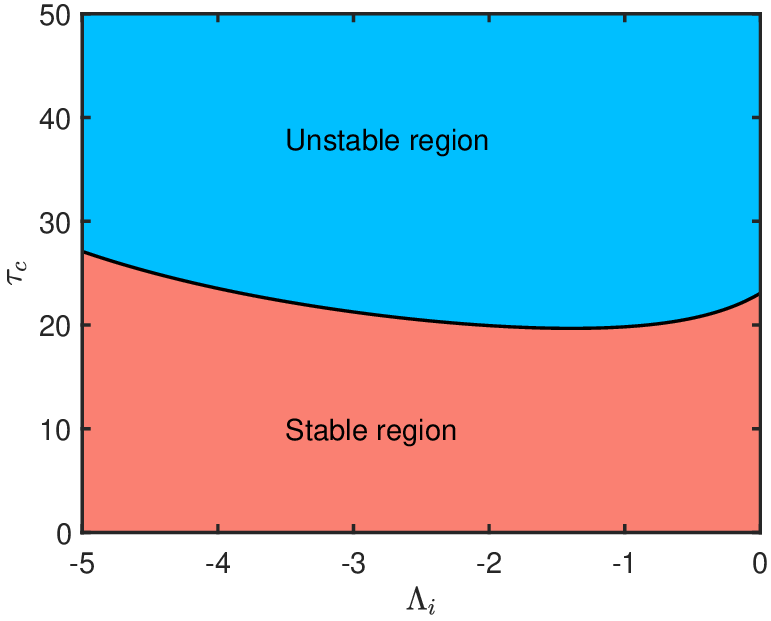}}
\caption{The critical value $\tau_c$ decreases as $\Lambda_i$ increases when $\Lambda=5$, $\mu=0.035$, $\nu=0.05$, $\gamma=0.2$, $\alpha=0.01$, $\beta=0.006$, $d_1=0.01$, and $d_2=0.08$. In this case, the deep sky blue area represents the unstable region, while the salmon area represents the stable region.}\label{immunization}
\end{figure}

To examine the pattern generation of fractional-order systems on a network with $N=100$ nodes and explore the effects of delay, network topology, and diffusion coefficients on the pattern, we have designed three additional sets of experiments. We have considered: different delays, $\tau=20$ in Fig.~\ref{1D pattern}(a,b) and Fig.~\ref{pattern}(a,b), and $\tau=40$ in Figs.~\ref{1D pattern}(c,d,e,f) and Fig.~\ref{pattern}(c,d,e,f); different average degrees, $\langle k\rangle=5$ in Fig.~\ref{1D pattern}(a,b,d,e) and Fig.~\ref{pattern}(a,b,d,e), and $\langle k\rangle=8$ in Fig.~\ref{1D pattern}(c,f) and Fig.~\ref{pattern}(c,f); and different diffusion coefficients, $d_1=0.01,d_2=0.02$ in Fig.~\ref{1D pattern}(a,c,d) and Fig.~\ref{pattern}(a,c,d), and $d_1=0.01,d_2=0.08$ in Fig.~\ref{1D pattern}(b,e,f) and Fig.~\ref{pattern}(b,e,f).

From the results obtained, we observe that spatial patterns emerge only when $\tau=40$, $d_1=0.01$, $d_2=0.08$, and $\langle k\rangle=5$, with instability in space, see Fig.~\ref{1D pattern}(e) and Fig.~\ref{pattern}(e). However, this phenomenon is irregular and distinct from the traditional Turing instability in space uniform stability, dividing into two parts (high and low abundance). In particular, as the delay increases, the system first undergoes a periodic oscillation state in time, then interacts with non-uniform oscillation in space due to the diffusion coefficient's influence at a specific time, resulting in irregular spatial non-uniform oscillation.

It is worth noting that, when we exclude the interference caused by delay-induced time-periodic oscillation and study only whether diffusion will result in pattern emergence, we do not observe pattern generation, see Fig.~\ref{1D pattern}(a,b) and Fig.~\ref{pattern}(a,b). We believe this is solely related to the SIRS model we study. Moreover, by observing Fig.~\ref{1D pattern}(e,f) and Fig.~\ref{pattern}(e,f), we also find that spatial non-uniform oscillation gradually disappears with the increase of the network's average degree. The time-period oscillation does not change with the network topology's variation, see Fig.~\ref{1D pattern}(c,d) and Fig.~\ref{pattern}(c,d).
\begin{figure}[h]
\centering
\subfigure[]
{\includegraphics[width=0.32\textwidth]{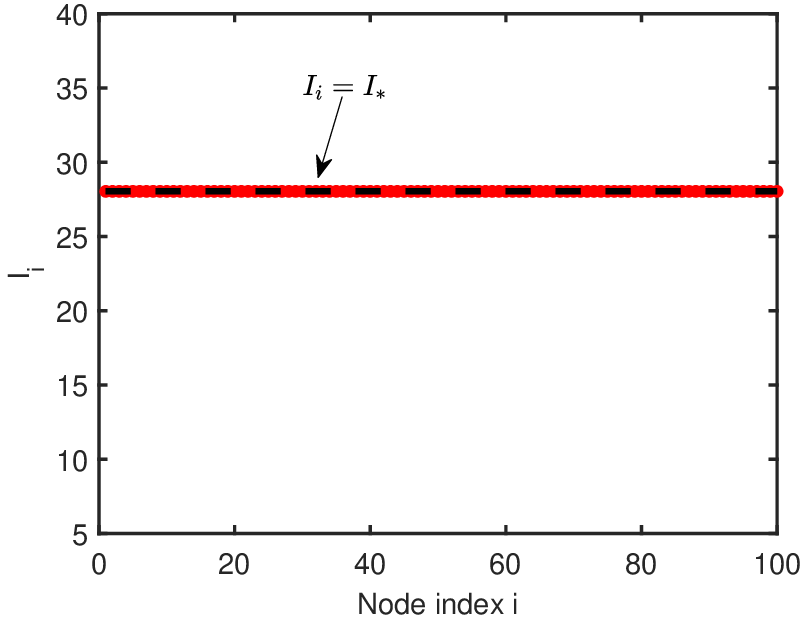}}
\subfigure[]
{\includegraphics[width=0.32\textwidth]{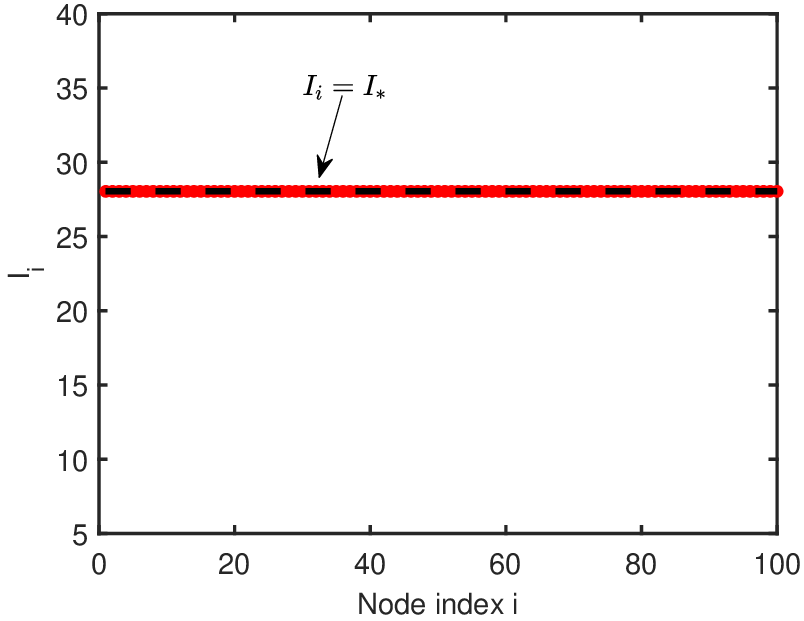}}
\subfigure[]
{\includegraphics[width=0.32\textwidth]{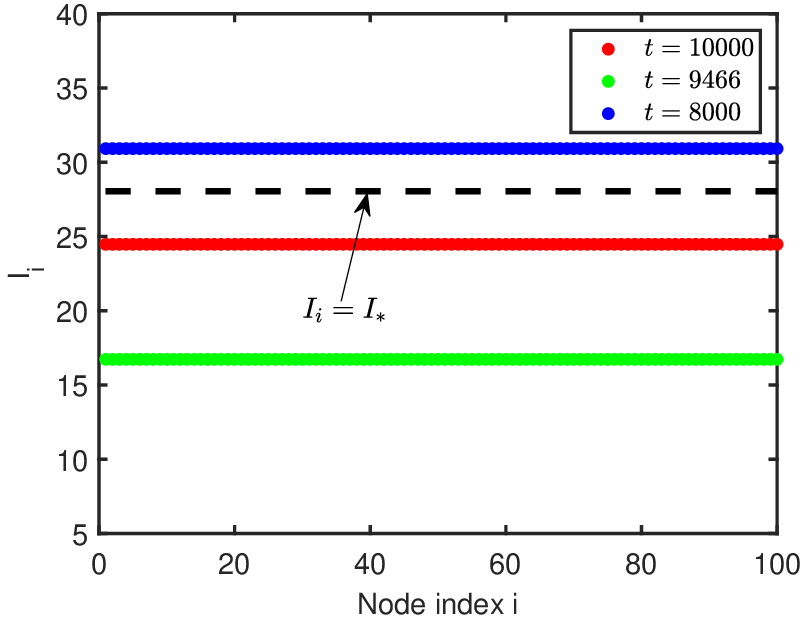}}
\subfigure[]
{\includegraphics[width=0.32\textwidth]{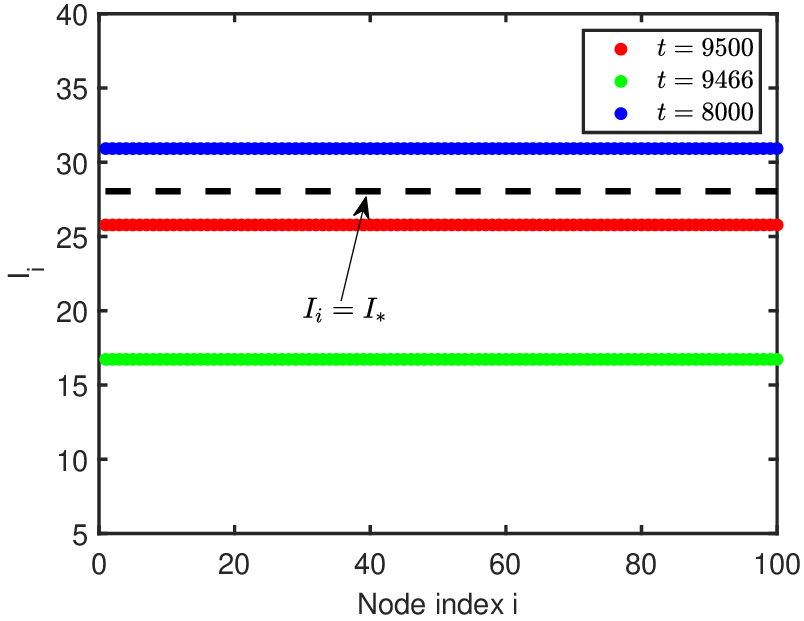}}
\subfigure[]
{\includegraphics[width=0.32\textwidth]{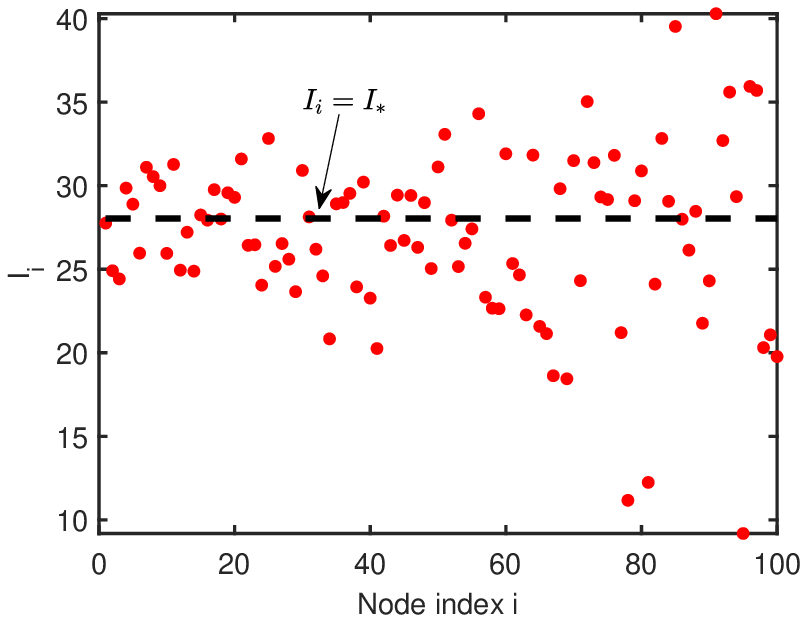}}
\subfigure[]
{\includegraphics[width=0.32\textwidth]{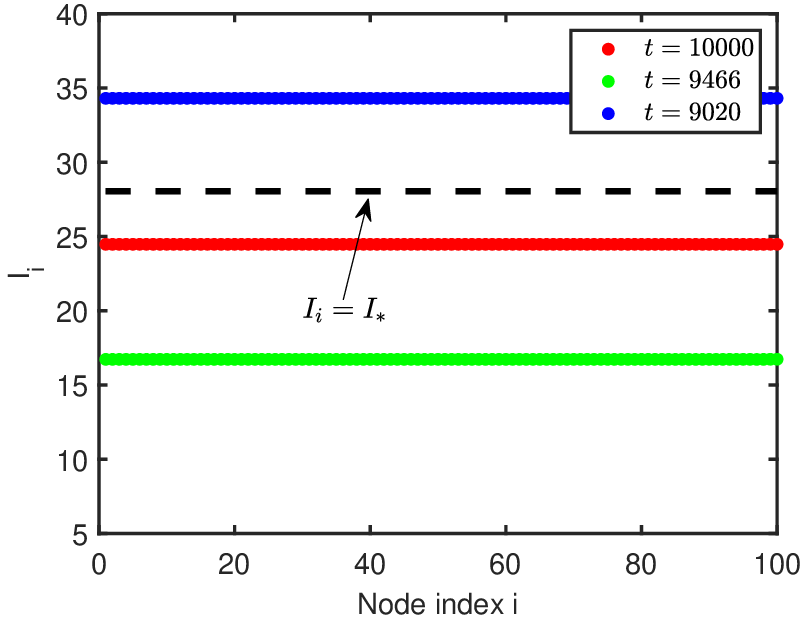}}
\caption{
Density of infected individuals as a function of node index $i$ for different values of delay $\tau$, average degree $\langle k\rangle$, and diffusion coefficients $d_1$ and $d_2$. (a)--(f) show the evolution of infected density with different nodes and times, represented by red, green, and blue dots. The black dotted line represents the value of the endemic disease equilibrium. The parameter values are as follows: (a) $\tau=20$, $d_1=0.01$, $d_2=0.02$, $\langle k\rangle=5$; (b) $\tau=20$, $d_1=0.01$, $d_2=0.08$, $\langle k\rangle=5$; (c) $\tau=40$, $d_1=0.01$, $d_2=0.02$, $\langle k\rangle=8$; (d) $\tau=40$, $d_1=0.01$, $d_2=0.02$, $\langle k\rangle=5$; (e) $\tau=40$, $d_1=0.01$, $d_2=0.08$, $\langle k\rangle=5$; and (f) $\tau=40$, $d_1=0.01$, $d_2=0.08$, $\langle k\rangle=8$.
}\label{1D pattern}
\end{figure}

\begin{figure}[h]
\centering
\subfigure[]
{\includegraphics[width=0.31\textwidth]{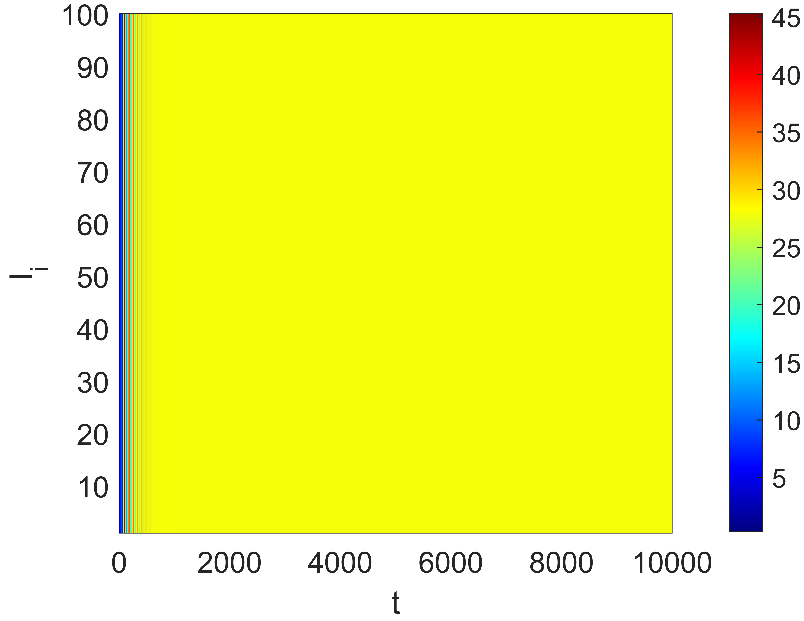}}
\subfigure[]
{\includegraphics[width=0.31\textwidth]{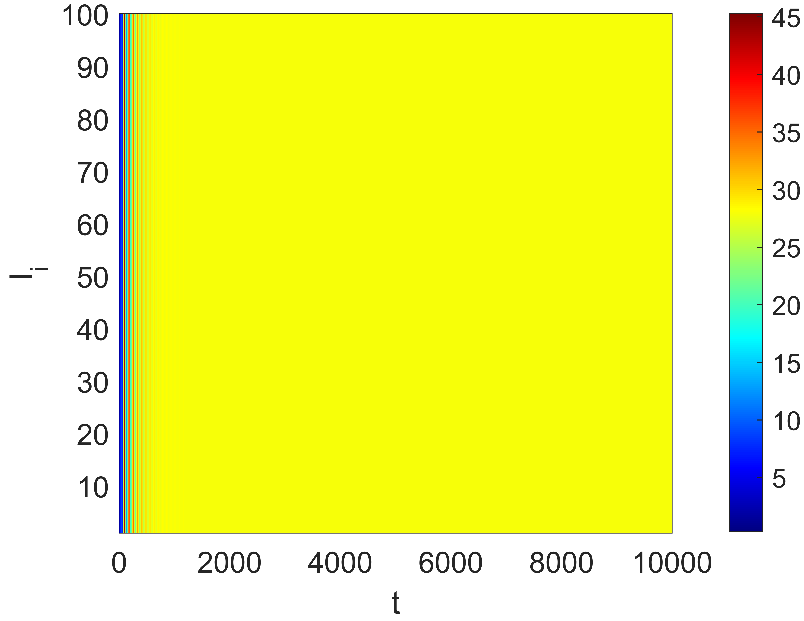}}
\subfigure[]
{\includegraphics[width=0.31\textwidth]{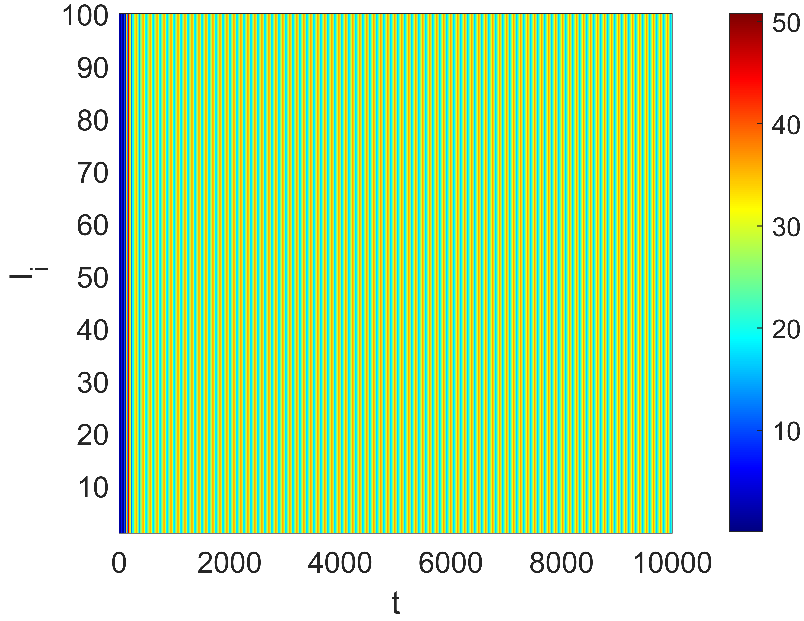}}
\subfigure[]{
\includegraphics[width=0.31\textwidth]{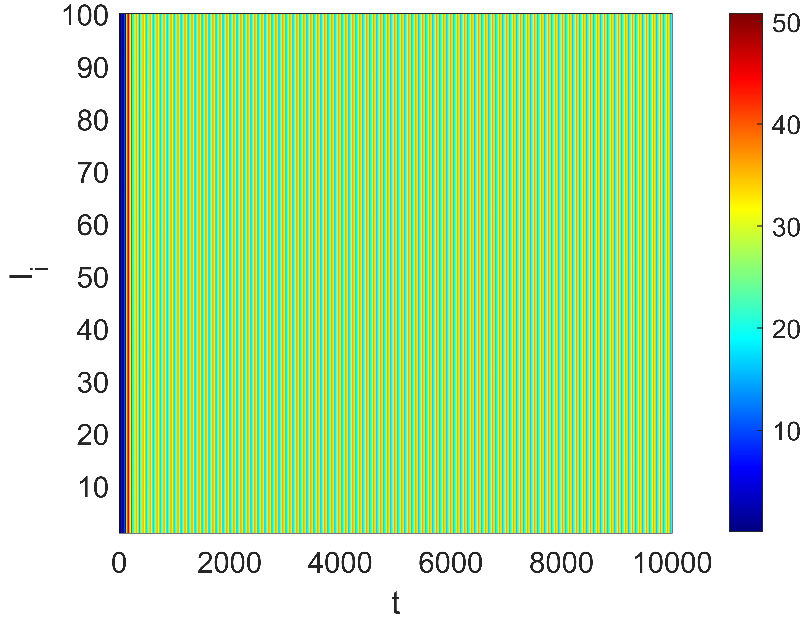}}
\subfigure[]{
\includegraphics[width=0.31\textwidth]{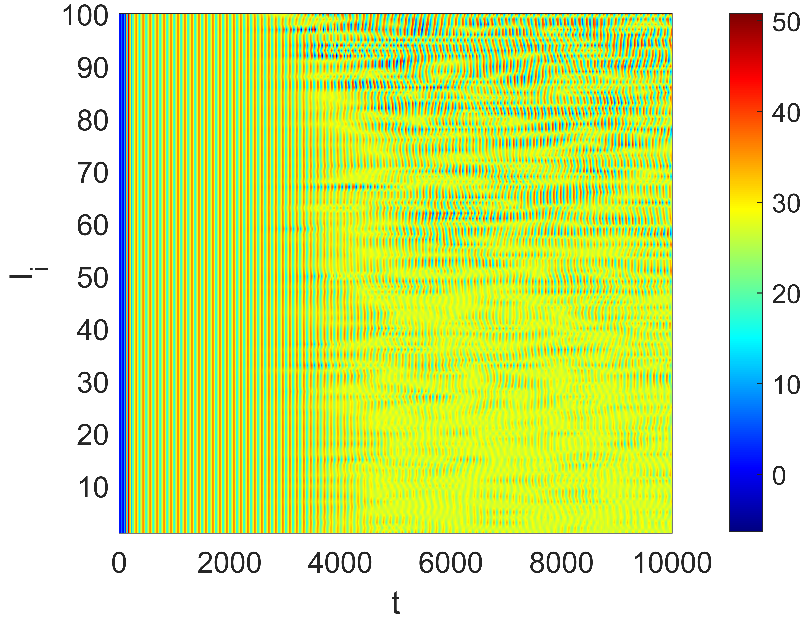}}
\subfigure[]{
\includegraphics[width=0.31\textwidth]{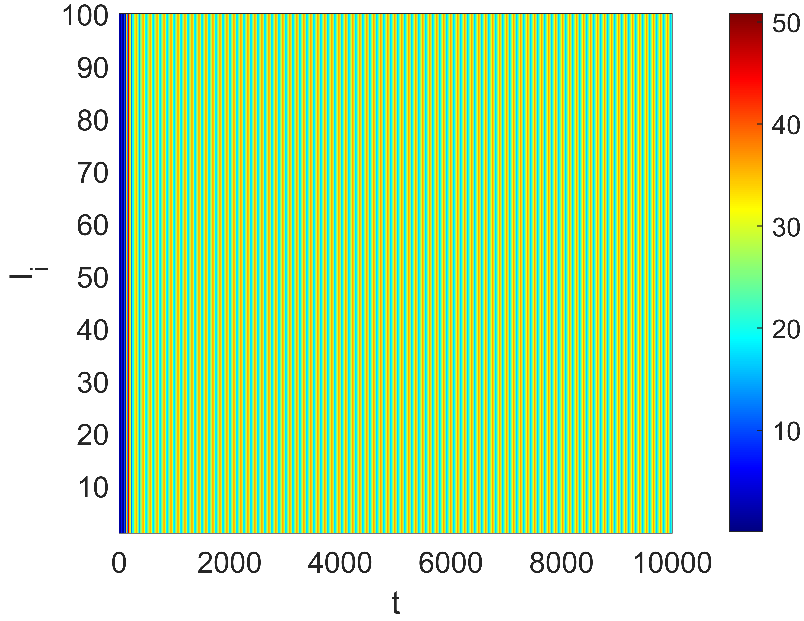}}
\caption{Spatiotemporal patterns of system \eqref{sys1}: (a) uniform spatial distribution when $\tau=20$, $d_1=0.01$, $d_2=0.02$, $\langle k\rangle=5$; (b) uniform spatial distribution when $\tau=20$, $d_1=0.01$, $d_2=0.08$, $\langle k\rangle=5$; (c) time periodic oscillation when $\tau=40$, $d_1=0.01$, $d_2=0.02$, $\langle k\rangle=8$; (d) time periodic oscillation when $\tau=40$, $d_1=0.01$, $d_2=0.02$, $\langle k\rangle=5$; (e) irregular spatiotemporal oscillation when $\tau=40$, $d_1=0.01$, $d_2=0.08$, $\langle k\rangle=5$; (f) time periodic oscillation when $\tau=40$, $d_1=0.01$, $d_2=0.08$, $\langle k\rangle=8$.}\label{pattern}
\end{figure}

We have conducted a comparative experiment to investigate the impact of fractional order on pattern generation, by designing two groups of experiments with $q=1$, see Fig.~\ref{alpha_1and0.95}(a,b,c), and $q=0.95$, see Fig.~\ref{alpha_1and0.95}(d,e,f), while keeping other parameters the same. The results show that, when $q=1$, the model generates spatiotemporal patterns, Fig.~\ref{alpha_1and0.95}(b), but when $q=0.95$, the spatiotemporal patterns disappeared, Fig.~\ref{alpha_1and0.95}(e). The density of infected individuals, $I_i$, on the ER random network over time and the curves of the maximum, minimum, and average values of the infected individual density across all nodes on the network over time are also shown in Fig.~\ref{alpha_1and0.95}(a,c,d,f) to support our observation. Note that the forward Euler method was used as the primary numerical method, with $\Delta t=1$, $T=20000$, $h=0.1$, and the 2D simulation regions $(x,y) \in \Omega=[0,4] \times[0,4]$ under Neumann boundary conditions. The Laplacian matrix was rewritten as a Laplace operator $\Delta$ in continuous media. Finally, Fig.~\ref{Initial value induced pattern}(a,b,c) show interesting patterns that appear due to the influence of the initial values, confirming that system~\eqref{sys1} also has spatiotemporal patterns in continuous media.

\begin{figure}[h]
\centering
\subfigure[]
{\includegraphics[width=0.31\textwidth ,height=1.6in]{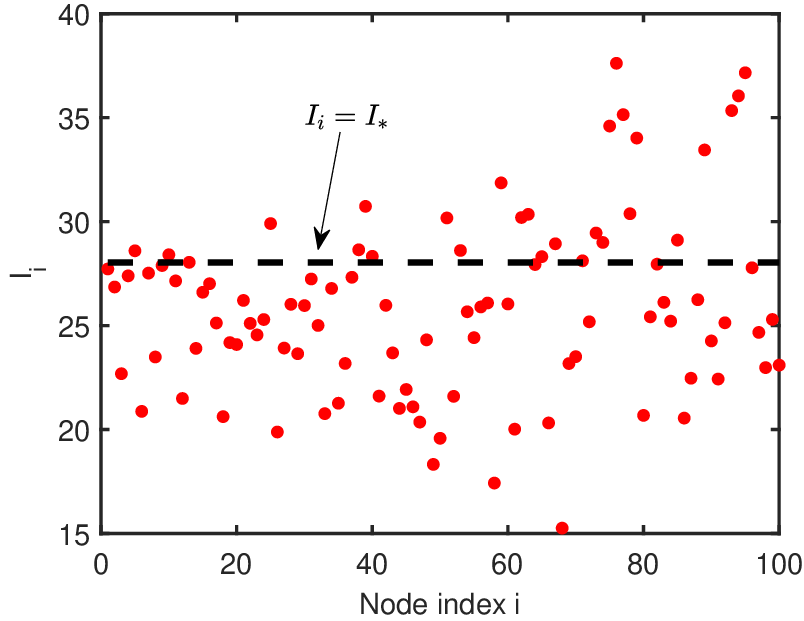}}
\subfigure[]
{\includegraphics[width=0.31\textwidth ,height=1.6in]{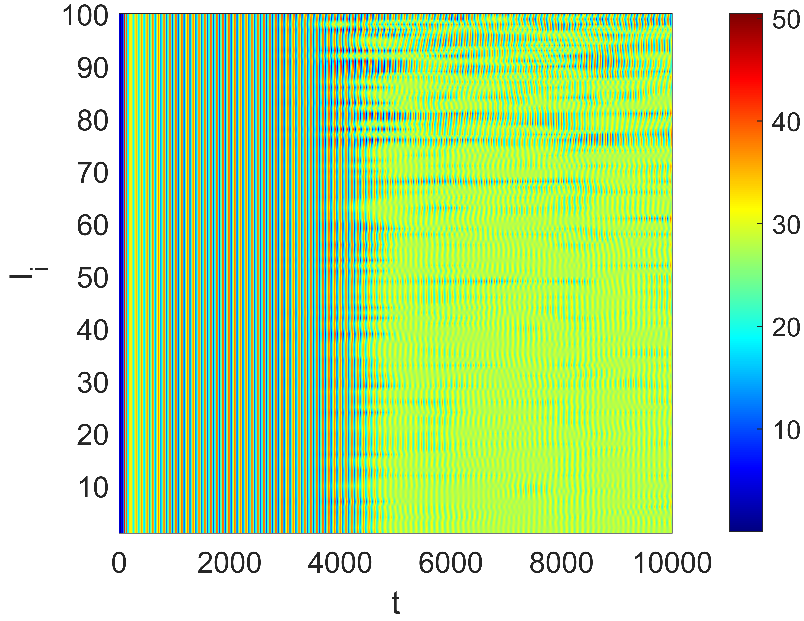}}
\subfigure[]
{\includegraphics[width=0.31\textwidth ,height=1.6in]{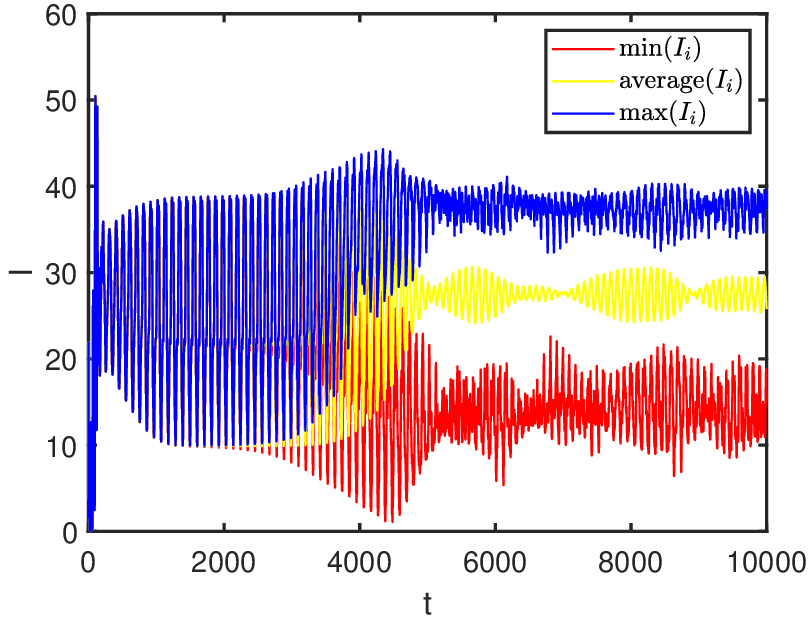}}
\subfigure[]{
\includegraphics[width=0.31\textwidth ,height=1.6in]{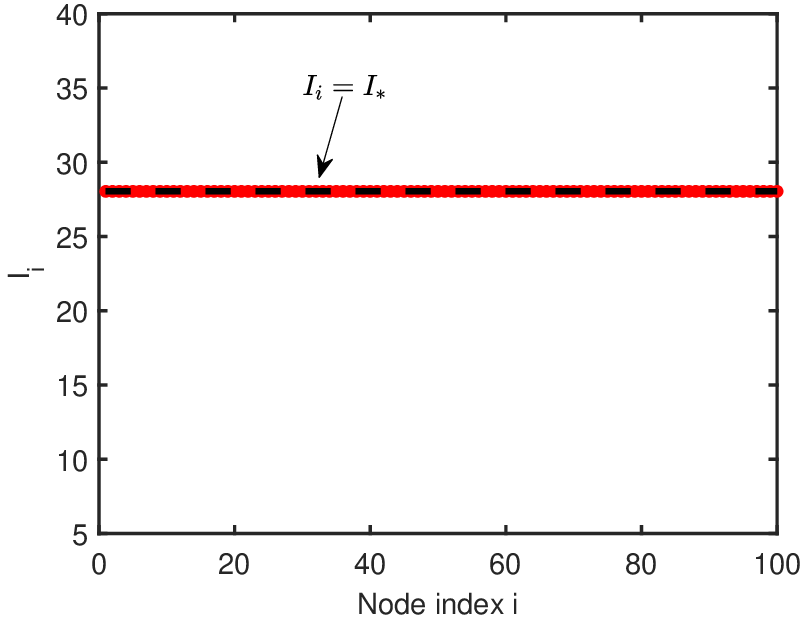}}
\subfigure[]{
\includegraphics[width=0.31\textwidth ,height=1.6in]{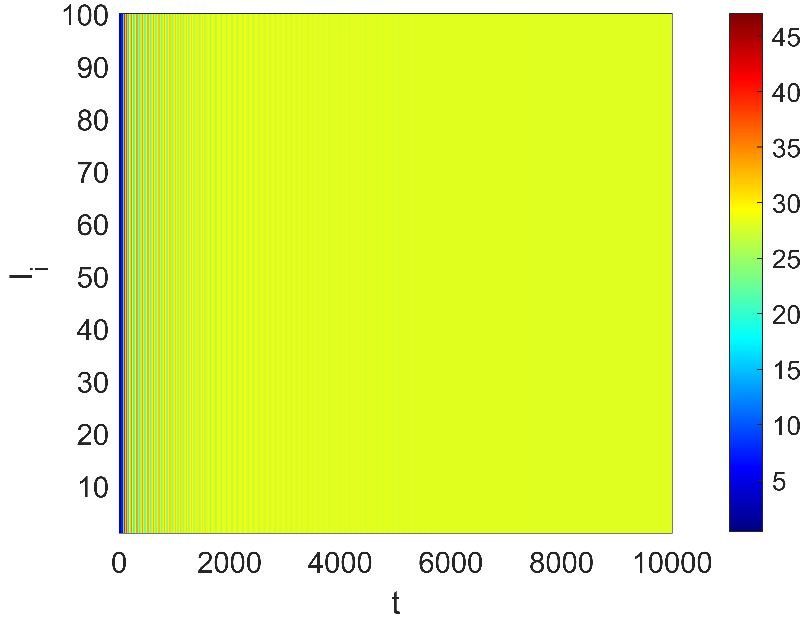}}
\subfigure[]{
\includegraphics[width=0.31\textwidth ,height=1.6in]{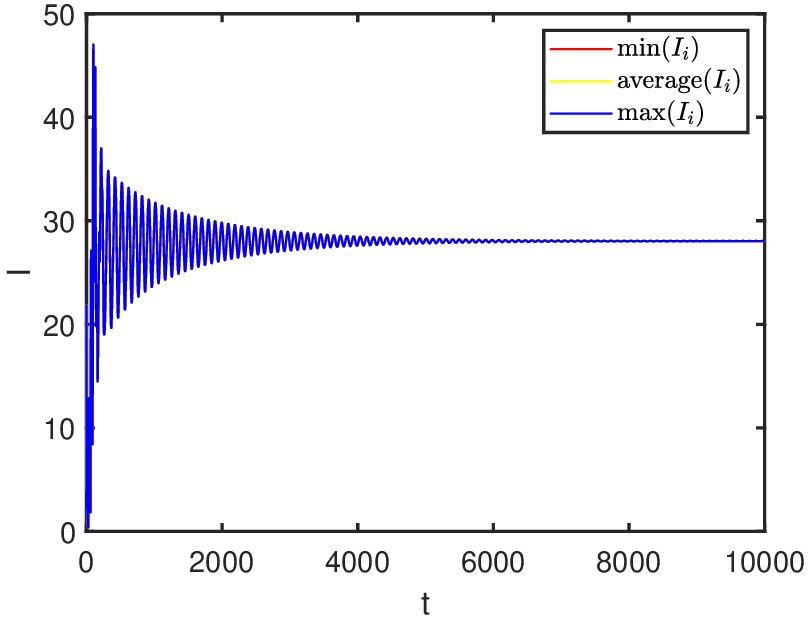}}
\caption{Spatiotemporal patterns of system \eqref{sys1} with different fractional orders. (a), (b), and (c) show spatiotemporal patterns generated when $\tau=30$, $d_1=0.01$, $d_2=0.05$, $\langle k\rangle=5$ and $q=1$. (d), (e), and (f) show no spatiotemporal patterns generated when $\tau=30$, $d_1=0.01$, $d_2=0.05$, $\langle k\rangle=5$ and $q=0.95$. Red dots represent the relationship between the density of infected individuals and node index $i$, while the black dotted line represents the endemic disease equilibrium.}\label{alpha_1and0.95}
\end{figure}

\begin{figure}[h]
\centering
\subfigure[]
{\includegraphics[width=0.32\textwidth]{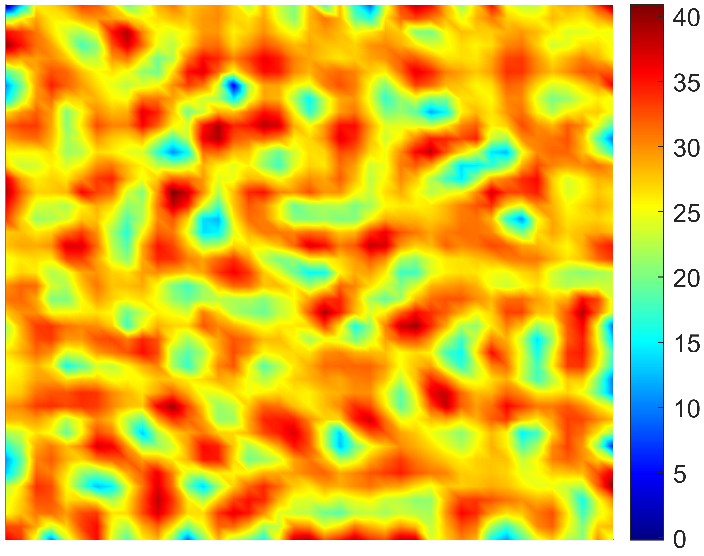}}
\subfigure[]
{\includegraphics[width=0.32\textwidth]{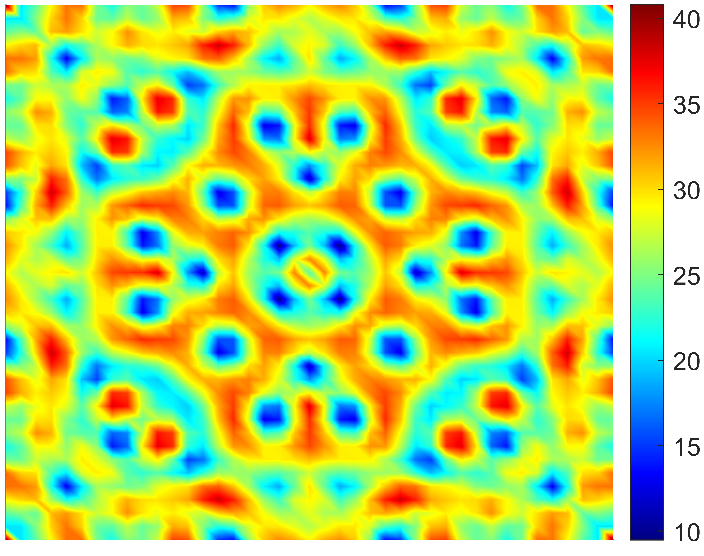}}
\subfigure[]
{\includegraphics[width=0.32\textwidth]{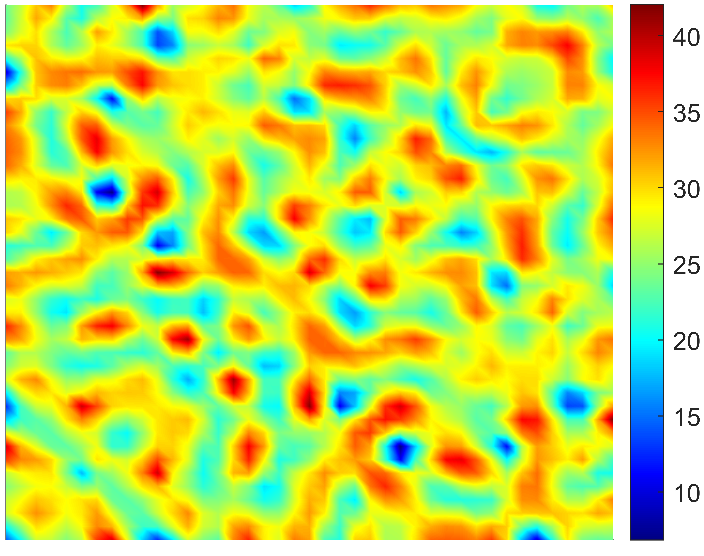}}
\caption{Initial value induced spatiotemporal patterns. When $t=20000$, for (a), $\left(S_0, I_0, R_0\right)=\left(S_*+0.01 \times \operatorname{rand}(0, 1), I_*+0.01 \times \operatorname{rand}(0, 1), R_*+0.01 \times \operatorname{rand}(0, 1)\right)$. For (b), $\left(S_0, I_0, R_0\right)=\left(S_*, I_*, R_*\right)+0.01 \times(1, 1, 1)$ if $(x-2)^2+(y-2)^2<2$, otherwise $\left(S_0, I_0, R_0\right)=\left(S_*, I_*, R_*\right)$. For (c), $\left(S_0, I_0, R_0\right)=\left(S_*, I_*, R_*\right)+0.01 \times(1, 1, 1)$ if $(x-1)^2+(y-1)^2<1$, $(x-1)^2+(y-3)^2<1$, $(x-3)^2+(y-1)^2<1$, or $(x-3)^2+(y-3)^2<1$, otherwise $\left(S_0, I_0, R_0\right)=\left(S_*, I_*, R_*\right)$. Other parameter values are $q=0.98$, $\tau=40$, $\Lambda=5$, $\mu=0.035$, $\nu=0.05$, $\gamma=0.2$, $\alpha=0.01$, $\beta=0.006$, $d_1=0.01$ and $d_2=0.08$.}\label{Initial value induced pattern}
\end{figure}

Finally, to better reflect the visualization of the main conclusions, we take several factors mainly considered in this paper, such as the network average degree, delay, and fractional order, as independent variables. With the help of numerical simulation, the evolution process of spatiotemporal patterns can be observed, as shown in Fig.~\ref{Evolution of spatiotemporal patterns} and Fig.~\ref{Diffusion introduced the evolution of spatiotemporal patterns}. Among them, from Fig.~\ref{Evolution of spatiotemporal patterns}(a) we can learn that, as the average degree increases, the spatial distribution of the population gradually becomes uniform, which means that the network average degree will inhibit the generation of spatial patterns. Similarly, when we fix the value of other parameters and change the delay parameter, it can be intuitively deduced that spatial patterns will appear as the delay increases, see Fig.~\ref{Evolution of spatiotemporal patterns}(b), which is also in good agreement with the theoretical analysis results. Considering the particularity of fractional-order systems, exciting phenomena occur when we change the order of fractional order. As the order increases, uniform spatial distribution is broken, resulting in spatiotemporal patterns. In other words, a decrease in the fractional order will inhibit the generation of spatiotemporal patterns. see Fig.~\ref{Evolution of spatiotemporal patterns}(c). In addition, we have also tested the evolution of spatiotemporal patterns with diffusion rate changes under two delays, $\tau=20$ for Fig.~\ref{Diffusion introduced the evolution of spatiotemporal patterns}(a), and $\tau=40$ for Fig.~\ref{Diffusion introduced the evolution of spatiotemporal patterns}(b). These results show that a single delay or diffusion rate effect does not lead to the generation of spatiotemporal patterns.
\begin{figure}[h]
\centering
\subfigure[]
{\includegraphics[width=0.32\textwidth]{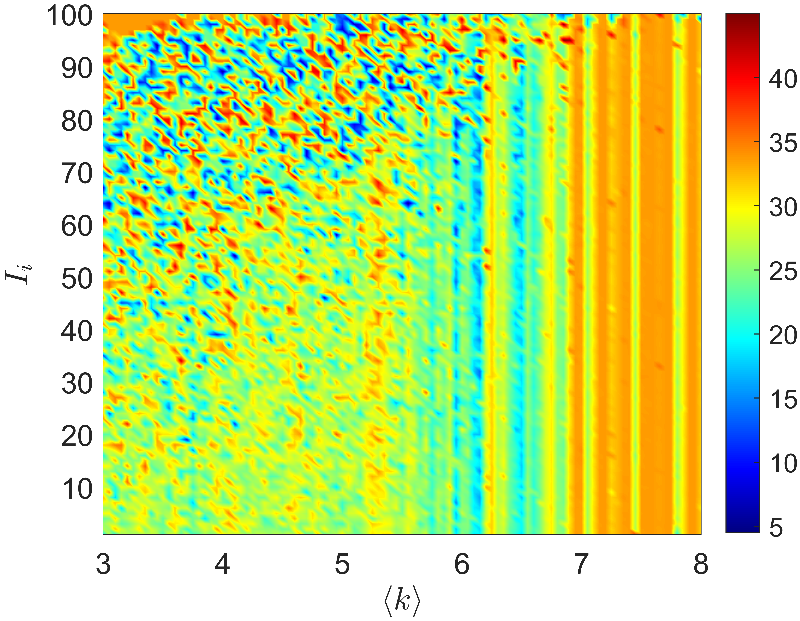}}
\subfigure[]
{\includegraphics[width=0.32\textwidth]{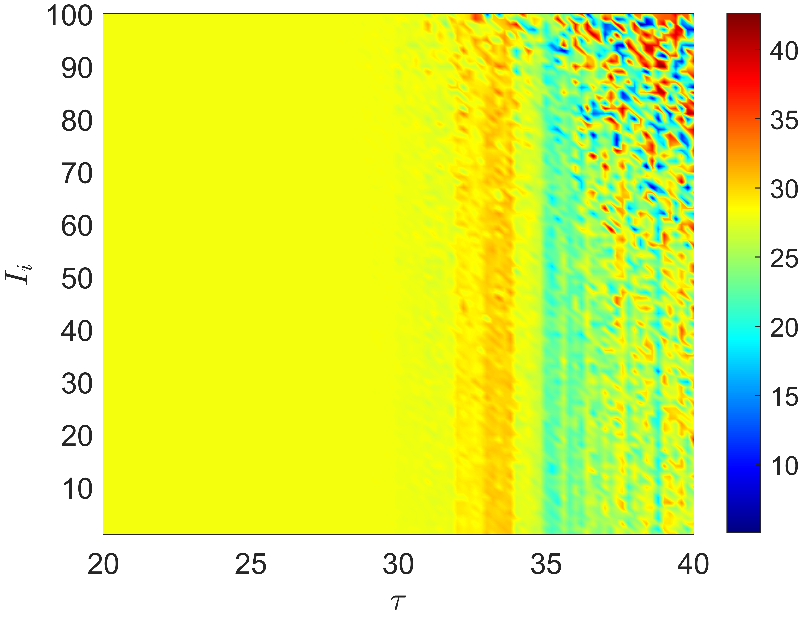}}
\subfigure[]
{\includegraphics[width=0.32\textwidth]{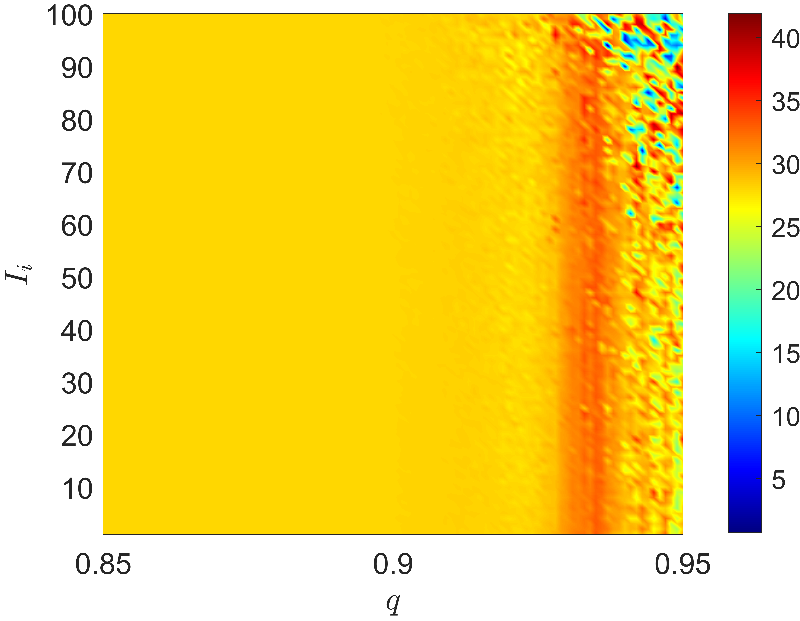}}
\caption{The evolution of spatiotemporal patterns with the average degree $\langle k \rangle$ for (a), delay $\tau$ for (b), and fractional order $q$ for (c), respectively. The values of the parameter are $\Lambda=5$, $\mu=0.035$, $\nu=0.05$, $\gamma=0.2$, $\alpha=0.01$, $\beta=0.006$, $d_1=0.01$, $\langle k \rangle=5$, and $d_2=0.08$.}\label{Evolution of spatiotemporal patterns}
\end{figure}

\begin{figure}[h]
\centering
\subfigure[]
{\includegraphics[width=0.45\textwidth]{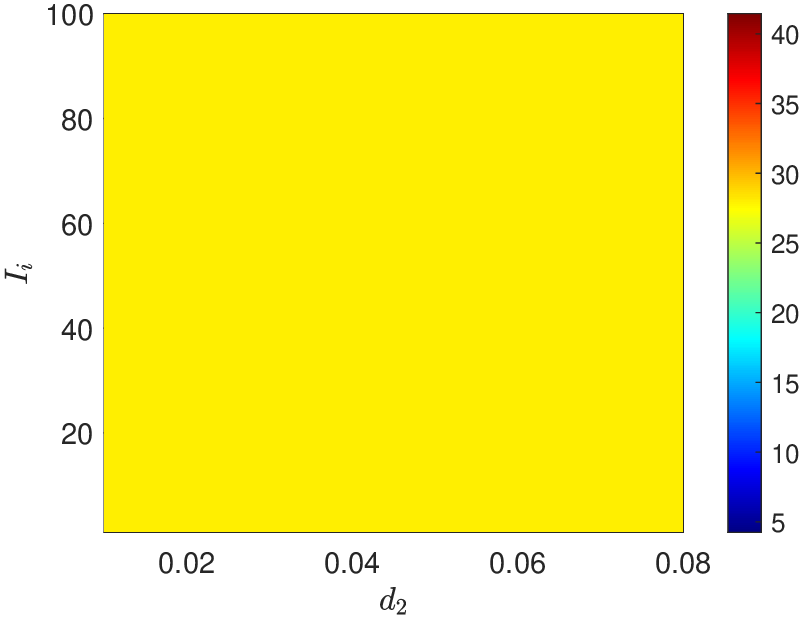}}
\subfigure[]
{\includegraphics[width=0.45\textwidth]{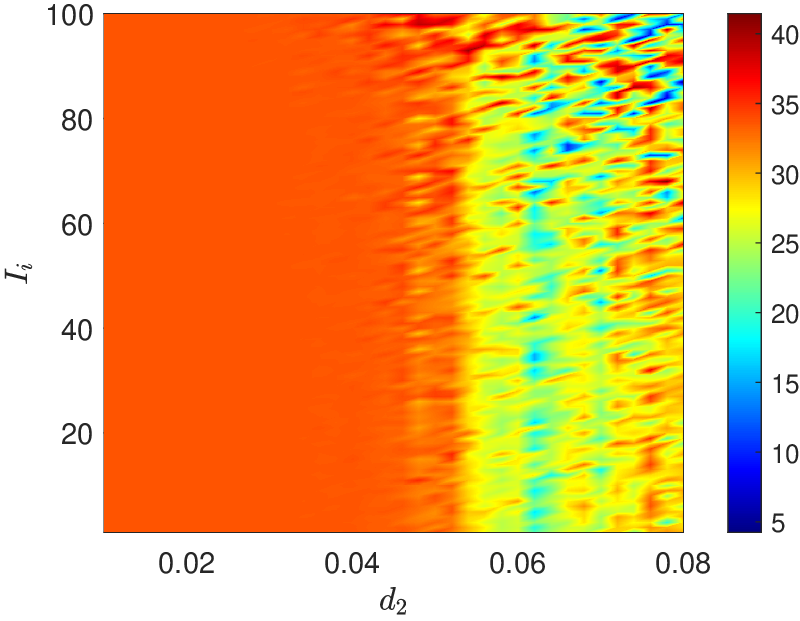}}
\caption{The evolution of spatiotemporal patterns under the combined effects of diffusion rate and delay, where $\tau=20$ for (a), $\tau=40$ for (b). The values of the parameter are $\Lambda=5$, $\mu=0.035$, $\nu=0.05$, $\gamma=0.2$, $\alpha=0.01$, $\beta=0.006$, $\langle k \rangle=5$, and $d_1=0.01$.}\label{Diffusion introduced the evolution of spatiotemporal patterns}
\end{figure}

\section{Conclusions}\label{section5}

Our results provide a new perspective for the research of delay-induced time-fractional order systems on networks. Among them, network topology, diffusion coefficient, and delay have essential effects on the excitation of the Turing pattern and the final differentiation of steady-state nodes. In the case of non-diffusion, the time-periodic oscillation phenomenon of the system is closely related to fractional order and delay. The reduction of fractional order promotes the stability of the system, while delay causes the instability of the system and leads to periodic oscillations. When we consider the diffusion term, with appropriate parameter values, the time-periodic oscillation phenomenon still exists and is not affected by the network topology. In particular, when the delay, diffusion coefficient, and average degree of the network are at appropriate values, an interesting phenomenon occurs, namely, irregular spatial non-uniform oscillation. Our explanation for this phenomenon is that the system first experiences a time-periodic oscillation state with the increase of delay, and then interacts with the non-uniform oscillation in space due to the effect of the diffusion coefficient at a particular time, resulting in the generation of irregular spatial non-uniform oscillation, i.e., spatiotemporal patterns.

It should be noted that the reason why fractional order inhibits the generation of spatiotemporal patterns is explained by the fact that fractional order affects temporal periodic oscillations, leading to the disappearance of the original temporal and spatial interaction, which does not alter the original uniform spatial distribution. Considering that the only network type studied in this paper is the Erd\H{o}s-R\'{e}nyi, we still know little about what effect the higher-order structure of the network has on the Turing pattern of the fractional-order system. Although some recent works have studied the pattern formation of networks with higher-order structures, most of them are based on integer-order systems~\cite{gao2023turing,muolo2023turing}. Therefore, in future research, we will study the pattern formation of the framework based on the fractional-order system and the higher-order structure of the network.

\section*{Declaration of competing interest}
The authors declare that they have no known competing financial interests or personal relationships that could have appeared to influence the work reported in this paper.

\section*{Data availability}
No data was used for the research described in this article.

\section*{Acknowledgements}
J.Z., Y.Y.\ and Y.Z.\ acknowledge support from the Nature Science Foundation of Guangdong Province (2020A1515010812 and 2021A1515011594). J.Z.\ acknowledges support by the scholarship from the China Scholarship Council (202106120290). Y.Y.\ acknowledges support by the scholarship from the China Scholarship Council (202206120230). A.A.\ and S.G.\ acknowledge support from Spanish Ministerio de Ciencia e Innovaci\'on (PID2021-128005NB-C21), Generalitat de Catalunya (2021SGR-00633) and Universitat Rovira i Virgili (2022PFR-URV-56). A.A.\ also acknowledges support from ICREA Academia, and the James S.\ McDonnell Foundation (220020325).
This work was carried out by Y.Y. and J.Z. during their tenure as visiting students at the Universitat Rovira i Virgili (URV).


\begin{thebibliography}{10}
\expandafter\ifx\csname url\endcsname\relax
  \def\url#1{\texttt{#1}}\fi
\expandafter\ifx\csname urlprefix\endcsname\relax\def\urlprefix{URL }\fi
\expandafter\ifx\csname href\endcsname\relax
  \def\href#1#2{#2} \def\path#1{#1}\fi

\bibitem{turing1952chemical}
A.~Turing, The chemical basis of morphogenesis, Philosophical Transactions of
  the Royal Society of London Series B 237~(641) (1952) 37--72.
\newblock \href {http://dx.doi.org/https://doi.org/10.1098/rstb.1952.0012}
  {\path{doi:https://doi.org/10.1098/rstb.1952.0012}}.

\bibitem{prigogine1968symmetry}
I.~Prigogine, R.~Lefever, Symmetry breaking instabilities in dissipative
  systems. ii, The Journal of Chemical Physics 48~(4) (1968) 1695--1700.
\newblock \href {http://dx.doi.org/https://doi.org/10.1063/1.1668896}
  {\path{doi:https://doi.org/10.1063/1.1668896}}.

\bibitem{castets1990experimental}
V.~Castets, E.~Dulos, J.~Boissonade, P.~De~Kepper, Experimental evidence of a
  sustained standing turing-type nonequilibrium chemical pattern, Physical
  Review Letters 64~(24) (1990) 2953.
\newblock \href {http://dx.doi.org/10.1103/PhysRevLett.64.2953}
  {\path{doi:10.1103/PhysRevLett.64.2953}}.

\bibitem{ouyang1991transition}
Q.~Ouyang, H.~L. Swinney, Transition from a uniform state to hexagonal and
  striped turing patterns, Nature 352~(6336) (1991) 610--612.
\newblock \href {http://dx.doi.org/https://doi.org/10.1038/352610a0}
  {\path{doi:https://doi.org/10.1038/352610a0}}.

\bibitem{asllanni22}
M.~Asllani, B.~A. Siebert, A.~Arenas, J.~P. Gleeson,
  \href{https://doi.org/10.1063/5.0060466}{{Symmetry-breaking mechanism for the
  formation of cluster chimera patterns}}, Chaos: An Interdisciplinary Journal
  of Nonlinear Science 32~(1), 013107.
\newblock \href {http://dx.doi.org/10.1063/5.0060466}
  {\path{doi:10.1063/5.0060466}}.
\newline\urlprefix\url{https://doi.org/10.1063/5.0060466}

\bibitem{sun2012pattern}
G.-Q. Sun, Pattern formation of an epidemic model with diffusion, Nonlinear
  Dynamics 69~(3) (2012) 1097--1104.
\newblock \href {http://dx.doi.org/https://doi.org/10.1007/s11071-012-0330-5}
  {\path{doi:https://doi.org/10.1007/s11071-012-0330-5}}.

\bibitem{chang2020cross}
L.~Chang, M.~Duan, G.~Sun, Z.~Jin, Cross-diffusion-induced patterns in an {SIR}
  epidemic model on complex networks, Chaos: An Interdisciplinary Journal of
  Nonlinear Science 30~(1) (2020) 013147.
\newblock \href {http://dx.doi.org/https://doi.org/10.1063/1.5135069}
  {\path{doi:https://doi.org/10.1063/1.5135069}}.

\bibitem{chang2022optimal}
L.~Chang, S.~Gao, Z.~Wang, Optimal control of pattern formations for an {SIR}
  reaction--diffusion epidemic model, Journal of Theoretical Biology 536 (2022)
  111003.
\newblock \href {http://dx.doi.org/https://doi.org/10.1016/j.jtbi.2022.111003}
  {\path{doi:https://doi.org/10.1016/j.jtbi.2022.111003}}.

\bibitem{chang2022sparse}
L.~Chang, W.~Gong, Z.~Jin, G.-Q. Sun, Sparse optimal control of pattern
  formations for an {SIR} reaction-diffusion epidemic model, SIAM Journal on
  Applied Mathematics 82~(5) (2022) 1764--1790.
\newblock \href {http://dx.doi.org/https://doi.org/10.1137/22M1472127}
  {\path{doi:https://doi.org/10.1137/22M1472127}}.

\bibitem{zheng2022pattern}
Q.~Zheng, V.~Pandey, J.~Shen, Y.~Xu, L.~Guan, Pattern dynamics in the epidemic
  model with diffusion network, Europhysics Letters 137~(4) (2022) 42002.
\newblock \href {http://dx.doi.org/10.1209/0295-5075/ac58bd}
  {\path{doi:10.1209/0295-5075/ac58bd}}.

\bibitem{zhou2022complex}
J.~Zhou, Y.~Zhao, Y.~Ye, Complex dynamics and control strategies of seir
  heterogeneous network model with saturated treatment, Physica A: Statistical
  Mechanics and its Applications 608 (2022) 128287.
\newblock \href {http://dx.doi.org/https://doi.org/10.1016/j.physa.2022.128287}
  {\path{doi:https://doi.org/10.1016/j.physa.2022.128287}}.

\bibitem{fernandes2012turing}
L.~D. Fernandes, M.~De~Aguiar, Turing patterns and apparent competition in
  predator-prey food webs on networks, Physical Review E 86~(5) (2012) 056203.
\newblock \href {http://dx.doi.org/10.1103/PhysRevE.86.056203}
  {\path{doi:10.1103/PhysRevE.86.056203}}.

\bibitem{zhang2014delay}
T.~Zhang, H.~Zang, Delay-induced turing instability in reaction-diffusion
  equations, Physical Review E 90~(5) (2014) 052908.
\newblock \href {http://dx.doi.org/10.1103/PhysRevE.90.052908}
  {\path{doi:10.1103/PhysRevE.90.052908}}.

\bibitem{liu2019pattern}
H.~Liu, Y.~Ye, Y.~Wei, W.~Ma, M.~Ma, K.~Zhang, Pattern formation in a
  reaction-diffusion predator-prey model with weak allee effect and delay,
  Complexity 2019 (2019) 6282958.
\newblock \href {http://dx.doi.org/https://doi.org/10.1155/2019/6282958}
  {\path{doi:https://doi.org/10.1155/2019/6282958}}.

\bibitem{othmer1971instability}
H.~G. Othmer, L.~Scriven, Instability and dynamic pattern in cellular networks,
  Journal of Theoretical Biology 32~(3) (1971) 507--537.
\newblock \href
  {http://dx.doi.org/https://doi.org/10.1016/0022-5193(71)90154-8}
  {\path{doi:https://doi.org/10.1016/0022-5193(71)90154-8}}.

\bibitem{othmer1974non}
H.~G. Othmer, L.~Scriven, Non-linear aspects of dynamic pattern in cellular
  networks, Journal of Theoretical Biology 43~(1) (1974) 83--112.
\newblock \href
  {http://dx.doi.org/https://doi.org/10.1016/S0022-5193(74)80047-0}
  {\path{doi:https://doi.org/10.1016/S0022-5193(74)80047-0}}.

\bibitem{horsthemke2004network}
W.~Horsthemke, K.~Lam, P.~K. Moore, Network topology and turing instabilities
  in small arrays of diffusively coupled reactors, Physics Letters A 328~(6)
  (2004) 444--451.
\newblock \href
  {http://dx.doi.org/https://doi.org/10.1016/j.physleta.2004.06.044}
  {\path{doi:https://doi.org/10.1016/j.physleta.2004.06.044}}.

\bibitem{moore2005localized}
P.~K. Moore, W.~Horsthemke, Localized patterns in homogeneous networks of
  diffusively coupled reactors, Physica D: Nonlinear Phenomena 206~(1-2) (2005)
  121--144.
\newblock \href {http://dx.doi.org/https://doi.org/10.1016/j.physd.2005.05.002}
  {\path{doi:https://doi.org/10.1016/j.physd.2005.05.002}}.

\bibitem{petit2017theory}
J.~Petit, B.~Lauwens, D.~Fanelli, T.~Carletti, Theory of turing patterns on
  time varying networks, Physical Review Letters 119~(14) (2017) 148301.
\newblock \href {http://dx.doi.org/10.1103/PhysRevLett.119.148301}
  {\path{doi:10.1103/PhysRevLett.119.148301}}.

\bibitem{zheng2020turing}
Q.~Zheng, J.~Shen, Y.~Xu, Turing instability in the reaction-diffusion network,
  Physical Review E 102~(6) (2020) 062215.
\newblock \href {http://dx.doi.org/10.1103/PhysRevE.102.062215}
  {\path{doi:10.1103/PhysRevE.102.062215}}.

\bibitem{muolo2023turing}
R.~Muolo, L.~Gallo, V.~Latora, M.~Frasca, T.~Carletti, Turing patterns in
  systems with high-order interactions, Chaos, Solitons \& Fractals 166 (2023)
  112912.
\newblock \href {http://dx.doi.org/https://doi.org/10.1016/j.chaos.2022.112912}
  {\path{doi:https://doi.org/10.1016/j.chaos.2022.112912}}.

\bibitem{nakao2010turing}
H.~Nakao, A.~S. Mikhailov, Turing patterns in network-organized
  activator--inhibitor systems, Nature Physics 6~(7) (2010) 544--550.
\newblock \href {http://dx.doi.org/https://doi.org/10.1038/nphys1651}
  {\path{doi:https://doi.org/10.1038/nphys1651}}.

\bibitem{sun2016pattern}
G.-Q. Sun, M.~Jusup, Z.~Jin, Y.~Wang, Z.~Wang, Pattern transitions in spatial
  epidemics: Mechanisms and emergent properties, Physics of life reviews 19
  (2016) 43--73.
\newblock \href {http://dx.doi.org/10.1016/j.plrev.2016.08.002}
  {\path{doi:10.1016/j.plrev.2016.08.002}}.

\bibitem{stancevic2013turing}
O.~Stancevic, C.~Angstmann, J.~M. Murray, B.~I. Henry, Turing patterns from
  dynamics of early hiv infection, Bulletin of mathematical biology 75 (2013)
  774--795.
\newblock \href {http://dx.doi.org/10.1007/s11538-013-9834-5}
  {\path{doi:10.1007/s11538-013-9834-5}}.

\bibitem{ye2021bifurcation}
Y.~Ye, Y.~Zhao, Bifurcation analysis of a delay-induced predator--prey model
  with allee effect and prey group defense, International Journal of
  Bifurcation and Chaos 31~(10) (2021) 2150158.
\newblock \href {http://dx.doi.org/https://doi.org/10.1142/S0218127421501583}
  {\path{doi:https://doi.org/10.1142/S0218127421501583}}.

\bibitem{ye2022promotion}
Y.~Ye, Y.~Zhao, J.~Zhou, Promotion of cooperation mechanism on the stability of
  delay-induced host-generalist parasitoid model, Chaos, Solitons \& Fractals
  165 (2022) 112882.
\newblock \href {http://dx.doi.org/https://doi.org/10.1016/j.chaos.2022.112882}
  {\path{doi:https://doi.org/10.1016/j.chaos.2022.112882}}.

\bibitem{zhou2022bifurcation}
J.~Zhou, Y.~Zhao, Y.~Ye, Y.~Bao, Bifurcation analysis of a fractional-order
  simplicial {SIRS} system induced by double delays, International Journal of
  Bifurcation and Chaos 32~(05) (2022) 2250068.
\newblock \href {http://dx.doi.org/https://doi.org/10.1142/S0218127422500687}
  {\path{doi:https://doi.org/10.1142/S0218127422500687}}.

\bibitem{lu2020fractional}
Z.~Lu, Y.~Yu, Y.~Chen, G.~Ren, C.~Xu, S.~Wang, Z.~Yin, A fractional-order
  seihdr model for covid-19 with inter-city networked coupling effects,
  Nonlinear dynamics 101~(3) (2020) 1717--1730.
\newblock \href {http://dx.doi.org/10.1007/s11071-020-05848-4}
  {\path{doi:10.1007/s11071-020-05848-4}}.

\bibitem{kilicman2018fractional}
A.~Kilicman, et~al., A fractional order sir epidemic model for dengue
  transmission, Chaos, Solitons \& Fractals 114 (2018) 55--62.
\newblock \href {http://dx.doi.org/10.1016/j.chaos.2018.06.031}
  {\path{doi:10.1016/j.chaos.2018.06.031}}.

\bibitem{higazy2020novel}
M.~Higazy, Novel fractional order sidarthe mathematical model of covid-19
  pandemic, Chaos, Solitons \& Fractals 138 (2020) 110007.
\newblock \href {http://dx.doi.org/10.1016/j.chaos.2020.110007}
  {\path{doi:10.1016/j.chaos.2020.110007}}.

\bibitem{zhang2020applicability}
Y.~Zhang, X.~Yu, H.~Sun, G.~R. Tick, W.~Wei, B.~Jin, Applicability of time
  fractional derivative models for simulating the dynamics and mitigation
  scenarios of covid-19, Chaos, Solitons \& Fractals 138 (2020) 109959.
\newblock \href {http://dx.doi.org/10.1016/j.chaos.2020.109959}
  {\path{doi:10.1016/j.chaos.2020.109959}}.

\bibitem{xu2020forecast}
C.~Xu, Y.~Yu, Y.~Chen, Z.~Lu, Forecast analysis of the epidemics trend of
  covid-19 in the usa by a generalized fractional-order seir model, Nonlinear
  dynamics 101~(3) (2020) 1621--1634.
\newblock \href {http://dx.doi.org/10.1007/s11071-020-05946-3}
  {\path{doi:10.1007/s11071-020-05946-3}}.

\bibitem{chang2019delay}
L.~Chang, C.~Liu, G.~Sun, Z.~Wang, Z.~Jin, Delay-induced patterns in a
  predator--prey model on complex networks with diffusion, New Journal of
  Physics 21~(7) (2019) 073035.
\newblock \href {http://dx.doi.org/10.1088/1367-2630/ab3078}
  {\path{doi:10.1088/1367-2630/ab3078}}.

\bibitem{gao2020cross}
S.~Gao, L.~Chang, X.~Wang, C.~Liu, X.~Li, Z.~Wang, Cross-diffusion on multiplex
  networks, New Journal of Physics 22~(5) (2020) 053047.
\newblock \href {http://dx.doi.org/10.1088/1367-2630/ab825e}
  {\path{doi:10.1088/1367-2630/ab825e}}.

\bibitem{du2013measuring}
M.~Du, Z.~Wang, H.~Hu, Measuring memory with the order of fractional
  derivative, Scientific Reports 3~(1) (2013) 1--3.
\newblock \href {http://dx.doi.org/https://doi.org/10.1038/srep03431}
  {\path{doi:https://doi.org/10.1038/srep03431}}.

\bibitem{djilali2020turing}
S.~Djilali, B.~Ghanbari, S.~Bentout, A.~Mezouaghi, Turing-hopf bifurcation in a
  diffusive mussel-algae model with time-fractional-order derivative, Chaos,
  Solitons \& Fractals 138 (2020) 109954.
\newblock \href {http://dx.doi.org/https://doi.org/10.1016/j.chaos.2020.109954}
  {\path{doi:https://doi.org/10.1016/j.chaos.2020.109954}}.

\bibitem{zhang2022impact}
N.~Zhang, Y.~Kao, B.~Xie, Impact of fear effect and prey refuge on a fractional
  order prey--predator system with beddington--deangelis functional response,
  Chaos: An Interdisciplinary Journal of Nonlinear Science 32~(4) (2022)
  043125.
\newblock \href {http://dx.doi.org/https://doi.org/10.1063/5.0082733}
  {\path{doi:https://doi.org/10.1063/5.0082733}}.

\bibitem{zheng2022turing}
Q.~Zheng, J.~Shen, Y.~Zhao, L.~Zhou, L.~Guan, Turing instability in the
  fractional-order system with random network, International Journal of Modern
  Physics B 36~(32) (2022) 2250234.
\newblock \href {http://dx.doi.org/https://doi.org/10.1142/S0217979222502344}
  {\path{doi:https://doi.org/10.1142/S0217979222502344}}.

\bibitem{kuznetsov2022robust}
M.~Kuznetsov, Robust controlled formation of turing patterns in three-component
  systems, Physical Review E 105~(1) (2022) 014209.
\newblock \href {http://dx.doi.org/10.1103/PhysRevE.105.014209}
  {\path{doi:10.1103/PhysRevE.105.014209}}.

\bibitem{podlubny1999199}
I.~Podlubny, Fractional differential equations, Academic Press, New York, 1999.

\bibitem{matignon1996stability}
D.~Matignon, Stability results for fractional differential equations with
  applications to control processing, in: Computational engineering in systems
  applications, Vol.~2, Lille, France, 1996, pp. 963--968.

\bibitem{deng2007stability}
W.~Deng, C.~Li, J.~L{\"u}, Stability analysis of linear fractional differential
  system with multiple time delays, Nonlinear Dynamics 48 (2007) 409--416.
\newblock \href {http://dx.doi.org/https://doi.org/10.1007/s11071-006-9094-0}
  {\path{doi:https://doi.org/10.1007/s11071-006-9094-0}}.

\bibitem{gao2023turing}
S.~Gao, L.~Chang, M.~Perc, Z.~Wang, Turing patterns in simplicial complexes,
  Physical Review E 107~(1) (2023) 014216.
\newblock \href {http://dx.doi.org/10.1103/PhysRevE.107.014216}
  {\path{doi:10.1103/PhysRevE.107.014216}}.

\end{thebibliography}

\end{document}